\documentclass[12pt, reqno]{amsart}

\IfFileExists{mymtpro2.sty}{%
  \usepackage[subscriptcorrection]{mymtpro2}
}{}

\usepackage{a4, upgreek, amssymb, amsmath}

\usepackage{appendix}

\newtheorem{theorem}{Theorem}[section]
\newtheorem{lemma}{Lemma}[section]
\newtheorem{corollary}{Corollary}[section]
\newtheorem{prop}{Proposition}[section]
\theoremstyle{definition}
\newtheorem{remark}[theorem]{Remark}

\newcommand{\labelnummer}{\mbox{\normalfont (\roman{numcount})}}%

\makeatletter

  {\let\curlabelspeicher\@currentlabel%
    \begin{list}{\labelnummer}%
      {\usecounter{numcount}\leftmargin0pt%
        \topsep0.5ex\partopsep2ex\parsep0pt\itemsep0ex\@plus1\p@%
        \labelwidth2.5em\itemindent3.5em\labelsep1em%
      }%
    \let\saveitem\item%
    \def\item{\saveitem%
      \def\@currentlabel{{\upshape\curlabelspeicher}$\,$\labelnummer}}%
    \let\savelabel\label%
    \def\label##1{\savelabel{##1}%
      \@bsphack%
        \ifmmode\else%
          \protected@write\@auxout{}%
          {\string\newlabel{##1item}{{\labelnummer}{\thepage}}}%
        \fi%
      \@esphack%
    }%
  }{\end{list}}%

\renewcommand{\appendix}{\def\thesection{\textsc{Appendix}}}

 \let\leq\le
 \let\geq\ge

\let\Re\undefined \let\Im\undefined
\DeclareMathOperator{\Re}{Re}
\DeclareMathOperator{\Im}{Im}

\DeclareMathOperator{\tr}{tr\kern1pt}

\makeatletter

\newif\ifper\pertrue
\def\per{.}

\def\bti{\@ifnextchar[\bbti\bbbti}
\def\bbti[#1]#2{#2, #1.}
\def\bbbti#1{#1.}

\def\z{\@ifnextchar[\zz\zzz}
\def\zz[#1]#2#3#4#5{\perfalse\emph{#2} \textbf{#3}, #4 (#5) [#1]}
\def\zzz#1#2#3#4{\emph{#1} \textbf{#2}, #3 (#4)\ifper\per\fi\pertrue}

\def\pub{\@ifstar\pubstar\pubnostar}
\def\pubnostar{\@ifnextchar[\@@pubnostar\@pubnostar}
\def\@@pubnostar[#1]#2#3#4{#2, #3, #4, #1\ifper\per\fi\pertrue}
\def\@pubnostar#1#2#3{#1, #2, #3\ifper\per\fi\pertrue}
\def\pubstar[#1]#2#3#4{\perfalse #2, #3, #4 [#1]\pertrue}

\makeatother

\newcommand{\bel}{\begin{equation} \label}
\newcommand{\ee}{\end{equation}}

\newcommand{\eps}{\epsilon}

\def\beq{\begin{equation}}
\def\eeq{\end{equation}}
\newcommand{\bea}{\begin{eqnarray}}
\newcommand{\eea}{\end{eqnarray}}
\newcommand{\beas}{\begin{eqnarray*}}
\newcommand{\eeas}{\end{eqnarray*}}

{

\newcommand{\hn}{\widetilde{H}_L^N}
\newcommand{\hm}{\widetilde{H}_L^M}
\renewcommand{\P}{\mathbb{P}}

\newcommand{\R}{\mathbb{R}}

\newcommand{\Z}{\mathbb{Z}}

\newcommand{\N}{\mathbb{N}}

\newcommand{\C}{\mathbb{C}}
\newcommand{\E}{\mathbb{E}}
\newcommand{\G}{\mathbb{G}}

\begin{document}

\title[Local eigenvalues statistics for random band matrices]{The density of states and local eigenvalue statistics for random band matrices of fixed width}

\author[B.\ Brodie]{Benjamin Brodie}
\address{Department of Mathematics,
    University of Kentucky,
    Lexington, Kentucky  40506-0027, USA}
\email{bbr239@g.uky.edu}

\author[P.\ D.\ Hislop]{Peter D.\ Hislop}
\address{Department of Mathematics,
    University of Kentucky,
    Lexington, Kentucky  40506-0027, USA}
\email{peter.hislop@uky.edu}


\begin{abstract}
We prove that the local eigenvalue statistics for $d=1$ random band matrices with fixed bandwidth and, for example, Gaussian entries, is given by a Poisson point process and we identify the intensity of the process.
The proof relies on an extension of the localization bounds of Schenker \cite{schenker} and the Wegner and Minami estimates.
These two estimates are proved using averaging over the diagonal disorder. The new component is a proof of the uniform convergence and the smoothness  of the density of states function. The limit function, known
to be the semicircle law with a band-width dependent error \cite{bmp,dps,dl,mpk}, is identified as the intensity of the limiting Poisson point process. The proof of these results for the density of states relies on a new result that simplifies and extends some of the ideas used by Dolai, Krishna, and Mallick \cite{dkm}. These authors proved regularity properties of the density of states for random Schr\"odinger operators (lattice and continuum) in the localization regime. The proof presented here applies to the random Schr\"odinger operators on a class of infinite graphs treated by in \cite{dkm} and extends the results of \cite{dkm} to probability measures with unbounded support.  The method also applies to fixed bandwidth RBM for $d=2,3$ provided certain localization bounds are known. 
\end{abstract}

\maketitle \thispagestyle{empty}

\tableofcontents


\section{Statement of the problem and results}\label{sec:introduction}
\setcounter{equation}{0}

Random band matrices (RBM) of size $2N+1$ and width $2L + 1$, for an integer $L$ with $0 \leq L \ll N$, provide an interesting model that exhibits a localization--delocalization transition when the width $2L+1$ increases as the size $N$ increases: If $L \sim N^\alpha$, then for $0 \leq \alpha < \frac{1}{2}$, the model is expected to be in the localized regime, whereas for $\frac{1}{2} < \alpha  \leq 1$, the model is expected to be delocalized. RBM of constant  band width $2L+1$ are in the localization phase. The eigenvectors are localized and the spectrum of the infinite-size matrix is pure point almost surely. In many ways, the constant width model has a large $N$ limit that resembles a random Schr\"odinger operator with Anderson-type potential on the line $\Z$. In this paper, we complete this analogy by proving that the local eigenvalue statistics at any fixed  energy $E_0$ is given by a Poisson point process with intensity measure $n_L^\infty(E_0) ~ds$. The function $n_L^\infty$ is the infinite $N$ density of states that is known to be given by the semicircle law up to $\mathcal{O}(L^{-1})$ corrections \cite{bmp} (see Remark \ref{remark:SclIds1}). In addition, we prove smoothness of the limiting density of states function $n_L^\infty(E)$ matching the regularity of the probability distribution.

Let $H_L^N$ be a $(2N+1)\times (2N+1)$ real, symmetric random band matrix with band width $2L + 1$ and $0 \leq L \ll N$. We define $H_L^N$ through its matrix elements:
\bea\label{eq:RBMdefn1}
\langle e_i,  H^N_L e_j \rangle & = & \frac{1}{\sqrt{2L+1}}\left\{\begin{array}{ccc}
v_{ij} & {\rm if} & |i-j| \leq L \\
0  & {\rm if} & |i-j| > L\end{array}\right. ,
\eea
with
$$
-N \leq i,j \leq N .
$$
The real random variables $v_{ij}= v_{ji}$ within the band are independent and identically distributed up to symmetry. Additional technical assumptions on the probability measure are given in Assumption 1 in section \ref{subsec:mainResult1}. 
These include the most common case of a Gaussian distribution. 

Much is known about RBM with power-law width $2L+1 = N^\alpha$, for $0 \leq \alpha \leq 1$. On the delocalization side, for $\frac{1}{2} < \alpha \leq 1$, in a series of three papers, Bourgade, Yin, Yang, and Yau \cite{byy1, byyy1, yy1} proved that if $\frac{3}{4} < \alpha \leq 1$, the eigenvectors of $H_L^N$ are extended with good probability, the limiting DOS is given by the semi-circle law, and that the $k$-point eigenvalue correlation functions for $H_L^N$ converge to those for the Gaussian orthogonal ensemble (GOE). Much less is known on the localization side $0 \leq \alpha < \frac{1}{2}$. Schenker \cite{schenker} proved that for Gaussian random variables and $0 \leq \alpha < \frac{1}{8}$, finite $N$ localization bounds hold for small moments of matrix elements of the resolvent in real energy intervals. A similar bound holds for more general random variables but the range of the exponent isn't specified. For the case of Gaussian random variables, the range of $\alpha$ was extended in \cite{psss} to $0 \leq \alpha < \frac{1}{7}$.

 Localization bounds, together with Wegner and Minami estimates, are among the main ingredients in the proof of the $N = \infty$ characterization of local eigenvalue statistics as a Poisson point process.
These three ingredients, localization bounds, the Wegner and Minami estimates,  necessary for a proof of the Poisson nature of the local eigenvalue statistics for RBM \eqref{eq:RBMdefn1}, have been known for some time now: localization bounds \cite{apsss,schenker} (although only in bounded energy intervals), the Wegner estimate \cite{psss}, and the Minami estimate \cite{psss}. What was missing in the proof of Poisson statistics was control over the density of states. In this paper, we prove smoothness of the density of states (depending on the smoothness of the probability measure) and uniform convergence of the local density of states  for the  endpoint case $\alpha = 0$ and for absolutely continuous probability measures with some regularity and that may have noncompact support. This corresponds to random band matrices with fixed bandwidth.  Our methods also require that we extend the localization bounds of Schenker \cite{schenker} to uniform bounds for all energies. We provide a simple proof of the Wegner and Minami estimates depending only on the diagonal randomness. This is sufficient for 
fixed bandwidth RBM but the constant in the bounds depends on the bandwidth $2L+1$. Controlling this constant is one of the advances of \cite{psss}. As a result, we also complete the proof of Poisson eigenvalue statistics for fixed bandwidth RBM. 


\subsection{The main results for fixed-width RBM}\label{subsec:mainResult1}

We summarize the main results of this paper. We begin with Assumption 1 on the random variables. 


\vspace{.1in}

\noindent
\textbf{Assumption 1:} \textit{The real random variables $v_{ij}$ in \eqref{eq:RBMdefn1} are independent, identically distributed (iid)
up to symmetry $v_{ij} = v_{ji}$. The common probability measure is absolutely continuous with respect to Lebesgue measure with a density $\rho \in L^\infty (\R)$ having support $I_\rho \subset \R$ that may be bounded or unbounded. There is an integer $k \geq 2$ so that the  probability density $\rho \geq 0$ satisfies 
\begin{enumerate}
\item regularity: $\rho \in C^k_\infty (\R)$, the set of $C^k$-functions with $\rho^{(m)} (x) \rightarrow 0$ as $|x| \rightarrow \infty$, for $0 \leq m \leq k$;
\item  integrability: the derivatives $\rho^{(m)} \in L^1 (\R)$, for $m=0,1, \ldots, k$;
\item finite moments: $\langle x \rangle^m \rho(x) \in L^1 (\R)$,  for $m=0, 1, \ldots, k+1$;
\item the Fourier transform $\widehat{\rho}$ satisfies the bound $\langle \xi \rangle^m \widehat{\rho}(\xi) \in L^\infty (\R)$,   for $m=0,1, \ldots, k$.
\end{enumerate}
}

All of the results of the paper hold under Assumption 1. Certain results hold under weaker conditions, for example, $\rho \in C^1$ suffices for the existence of the DOSf.  Conditions (2)-(4) are 
needed in the proof of the uniform convergence of the local integrated density of states \eqref{eq:lids1}. 
Assumption 1 allows probability densities $\rho$ with unbounded supports. 
The main case of interest is the Gaussian distribution for which
\beq\label{eq:gauss1}
\rho(x) = \frac{1}{\sqrt{2 \pi}}  e^{- \frac{1}{2} x^2} ,
\eeq 
and that satisfies these conditions with $k = \infty$. 

Let $\{E_L^N(j)\}_{j=-N}^N$ denote the set of the $2N+1$ eigenvalues of $H_L^N$. To study the local eigenvalue statistics for $H_L^N$ around $E_0\in\R$, we define the re-scaled eigenvalues
\beq\label{eq:rescaledEV1intro}
\widetilde{E}_L^N(j):=(2N+1)\left({E}_L^N(j) - E_0\right).
\eeq
The {\it eigenvalue point process $\xi_{N,L}^\omega$ for $H_L^N$ centered at $E_0$} is a random point measure supported on the re-scaled eigenvalues:
\beq\label{eq:localPP1}
\xi_{N,L}^\omega (s) ~ d s := \sum_{j=-N}^N \delta ( {\widetilde{E}}_L^N(j) - s) ~ d s.
\eeq

The density of states function $n_L^\infty (E)$ is defined for RBM in section \ref{sec:dos1}.
From the local eigenvalue counting function we form the  
  {\it local integrated density of states} ($\ell$IDS),  the expectation of the normalized eigenvalue counting function:
\beq\label{eq:lids1intro}
 N_L^N(E) = :  \frac{1}{2N+1} \, \E \{  \#\{ j \, : \, E_L^N(j)\leq E \} \} .
\eeq
The IDS is the limit of the $\ell$IDS as $N \rightarrow \infty$:
\beq\label{eq:ids1}
 N_L ^\infty(E) = : \lim_{N \rightarrow \infty}  \frac{1}{2N+1} \, \E \{  \#\{ j \, : \, E_L^N(j)\leq E \} \} .
\eeq
We will prove that this function is differentiable (provided the density $\rho$ is also) and the derivative of the $\ell$IDS is the {\it local density of states function} ($\ell$DOSf):
\beq\label{eq:localDosf11}
n_L ^N(E) =: \frac{d}{dE} N_L^N (E).
\eeq
We will prove that this converges uniformly to a function $n_L^\infty(E)$ that is the DOSf for $H_L^\infty$.

\begin{theorem}\label{thm:RBMmain1}
Let $H_L^N$ be a random band matrix with fixed-width $2L+1$ as in \eqref{eq:RBMdefn1}. The random variables $v_{ij}$, for $-N \leq i < j \leq N$ satisfy Assumption 1. The density of states function $n_L^\infty (E) \in C^{k-1}(\R)$. For any $E \in \R$, the local eigenvalue statistics, defined as the point process obtained from the weak limit  of \eqref{eq:localPP1},  is a Poisson point process with intensity measure $n_L^\infty (E) ds$. 
\end{theorem}


\subsection{Applications to random Schr\"odinger operators on infinite graphs}\label{subsec:RSOappl}

A theorem similar to Theorem \ref{thm:RBMmain1} applies to the local eigenvalue statistics and the DOSf for discrete random Schr\"odinger operators on a wide variety of infinite graphs $\G$, including the lattice $\Z^d$, as described in \cite{dkm} for energies in the complete localization regime $\Sigma^{\rm CL}$. The family of graphs $\G$ are infinite graphs with a metric $d_\G$ and 
with the property that for any $k_0 \in \G$, if $\Lambda_L := \{ k \in \G ~|~ d_\G ( k_0, k ) \leq L \}$, then $|\Lambda_L | = \mathcal{O}(L^\alpha)$, for some $\alpha > 0$. 

The regularity part of the following theorem was proven in \cite[Theorem 3.4]{dkm} for compactly-supported probability measures with an additional technical assumption that the $k^{th}$-derivative of the probability density is H\"older continuous due to a constraint in \cite[Theorem 2.2]{dkm}. 
The proof of the Poisson statistics  for lattices $\Z^d$ was given in \cite{minami}. The techniques developed in sections \ref{sec:fracMoment1} and \ref{sec:localBound1} applied to RSO yield the following theorem removing the technical constraint on the probability density and the necessity that the probability density have compact support. We also remark that the random perturbations need not be rank one but may be of uniform finite rank as treated in \cite{dkm} and \cite[chapters 6-7]{brodie1} on the Wegner $N$-orbital model. We summarize this in the following theorem stated for the rank one case.   

\begin{theorem}\label{thm:RSOmain1}
Let $H_\omega^L$ be the restriction of a random Schr\"odinger operator $H_\omega$ on $\ell^2 (\G)$ to the finite set $\Lambda_L \subset \G$ for which the random variables $\{ \omega_j ~|~ j \in \Z^d \}$ satisfy Assumption 1.  
The local integrated density of states $N_L(E)$ converges uniformly to the IDS $N(E)$ for $E \in \Sigma^{\rm CL}$. 
The density of states function $n (E) \in C^{k-1}(\R)$.  For any $E \in \Sigma^{CL}$, the local eigenvalue statistics, defined as the point process obtained from the weak limit  of \eqref{eq:localPP1},  is a Poisson point process with intensity measure $n (E) ds$. 
\end{theorem}


\subsection{Applications to fixed-width RBM in dimensions $d=2,3$}\label{subsec:highDimRBM}

As discussed above, for RBM in $d=1$, the critical bandwidth is conjectured to be $N^{\frac{1}{2}}$: RBM with bandwidths growing like $N^\alpha$ for $\frac{1}{2} < \alpha \leq 1$ are expected to be delocalized, whereas those with $0 \leq \alpha < \frac{1}{2}$ are expected to be localized. Consequently, fixed bandwidth RBM should (and do) exhibit localization at all energies.  A similar result should hold for RBM in $d=2$ as the conjectured critical bandwidth is $(\log N)^{\frac{1}{2}}$. However, for $d=3$, the conjectured critical bandwidth is $\mathcal{O}(1)$ (see, for example, \cite{bourgade}). The techniques developed in this paper apply to higher-dimensional fixed bandwidth RBM at energies in the localization regime if it is nonempty. For $d=2$, the conjectured critical bandwidth implies that localization holds for fixed bandwidth RBM. Provided a localization estimate of the type described in Theorem \ref{thm:fractMoment1} holds for these energies, the methods of this paper prove that the
limiting density of states is smooth (the degree of regularity depending on the regularity of the density of the probability measure) and that local eigenvalue statistics are Poisson point processes with the intensity given by the associated density of states. It is not clear if there is any localization regime for RBM in $d=3$, but if there is, and similar localization bounds hold for energies in that region, then the methods of this paper provide the same results. 

The DOS for fixed bandwidth RBM with Gaussian distributions in $d=3$ was studied by Disertori, Pinson, and Spencer \cite{dps} and for $d=2$ by Disertori and Lager \cite{dl}. These authors proved that the limiting DOS at all energies  in the bulk $(-2.2)$  is given by the semicircle distribution up to terms $\mathcal{O}(L^{-2})$. The method does not rely upon localization and is based on the supersymmetric representation of the finite region Green's function.


\subsection{Contents of the paper}\label{subsec:contents1}

The basic bounds for finite $N$ matrices are proven in section \ref{sec:fundamentals1} based on a spectral averaging theorem using only the diagonal randomness and the Schur complement formula. Spectral averaging is used to derive a Wegner and Minami estimate. The localization bound of Schenker \cite{schenker} is extended to complex energies using a high-energy Aizenman-Molchanov bound for RBM and subharmonicity. In section \ref{sec:dos1}, basic results on the density of states (DOS) for $N$ finite and infinite are derived. The smoothness of the DOS is proven following the ideas of Dolai, Krishna, and Mallick \cite{dkm}. The Schur complement formula plays an essential role in our simplification of some of the arguments of \cite[Theorem 2.2]{dkm} and in our extension to probability measures with non-compact support.   Section \ref{sec:fracMoment1} and \ref{sec:localBound1} present the main technical results: The reduction of the difference of the expectation of random variables related to $H_L^N$ and $H_L^M$ to finite moments of resolvents and, the bounding of these by localization bounds. 
 The convergence of the DOS functions and identification of the limit is given in section \ref{sec:ptwConvDOS1}.
Finally, the proof of the local eigenvalue statistics as a Poisson point process is given in section \ref{sec:poisson1}. The paper concludes with two appendices: In the first, section \ref{app:identities1}, we present some basic identities used in the fractional moment bounds, and in the second, section \ref{sec:loc1}, we present the proof of the high-energy localization bound for RBM by a modified Aizenman-Molchanov argument.  A note on notation: The value of various constants may change from line to line. Important dependencies are indicated by subscripts such as $C_{L,s}$, etc. 

\subsection{Acknowledgements} We thank P.\ Bourgade, D.\ Dolai, M.\ Krishna, and J.\ Schenker for stimulating discussions and helpful remarks.  Some of the results in this article are from the University of Kentucky doctoral dissertation of the first author \cite{brodie1}.


\section{Fundamentals: Spectral averaging and local bounds}\label{sec:fundamentals1}
\setcounter{equation}{0}

In this section, we treat general real symmetric $N \times N$ matrices $H_\omega^N$ and present the basic bounds on these random matrices $H^N_\omega$ depending only on the diagonal randomness. As mentioned in the introduction, the Wegner and Minami bounds obtained in the manner suffice for fixed bandwidth RBM. For the case of when the bandwidth increases with $N$, these constants in these estimates also grow. The results of \cite[Theorem 3]{psss} for Gaussian probability measures show that these constants may be chosen independent of $L$ and $N$.   
We begin with a version of rank-one and rank-two spectral averaging based on the Schur complement formula. We then present the Wegner and Minami estimates. These are derived using only the diagonal disorder. 
We also mention more refined estimates of Peled, Schenker, Shamis, Sodin \cite{psss}.
Finally, we derive an extension of Schenker's localization bound for random band matrices valid for complex energies.


%
%
%


\subsection{Spectral averaging}\label{subsec:SpectralAve1}

We need the following {\it a priori} bound on the expectation of the Green's functions of matrices with \emph{iid} random variables along the diagonal.

\begin{prop}[Spectral Averaging] \label{proposition:SpectralAve1}
Let $H_\omega^N$ be a self-adjoint $N \times N$-matrix such that its diagonal entries $[H_\omega^N]_{jj}:= \omega_{j}$  are $iid$ random variables with a common bounded probability density function $\rho$ vanishing at infinity. 
Let $\{ e_j ~|~ j = 1, \ldots, N \}$ be the standard orthonormal basis of $\C^N$.
\begin{enumerate}
\item[{(i.)}] For any $0<s<1$, there exists a finite constant $C_{\rho, s} > 0$, independent of $N$ and indices $j, k \in \{1, \ldots, N \}$, so that for any $z\in\C$, we have
\beq\label{eq:matrixEle1}
\E \left\{ \left| \left\langle e_j, \left(H_\omega^N -z\right)^{-1} e_k \right\rangle \right|^s  \right\} \leq C_{\rho, s} .
\eeq
\item[{(ii.)}] There exists a finite constant $C_\rho > 0$, independent of $N$ and $j \in \{ 1, \ldots, N \}$, so that for any $z\in\C \backslash \R$,
\beq\label{eq:ImMatrixEle1}
\E \left\{  \left| \Im \left\langle e_j, \left(H_\omega^N - z \right)^{-1} e_j \right\rangle \right| \right\} \leq C_\rho .
\eeq
\end{enumerate}
\end{prop}

\noindent
We begin with a lemma that is key to obtaining \eqref{eq:matrixEle1}. Some of the methods are from the proof of \cite[section 5, Theorem 7]{schenker}.

\begin{lemma}\label{lemma:SpectralAve2}
Let $\displaystyle V=\left(\begin{array}{cc} v_1 & 0 \\ 0 & v_2 \\ \end{array} \right)$ be a real diagonal random matrix with $v_1$ and $v_2$ independent random variables with common density $\rho$, and let $A$ be a self-adjoint $2 \times 2$ matrix independent of $V$. Then,
for all $t > 0$
\beq\label{eq:l1est1}
 \P \left\{ \| (V+A)^{-1} \| > t \right\} < \frac{4 \pi \| \rho \|_\infty}{t}.
\eeq
\end{lemma}

\begin{proof}
For any self-adjoint operator $H$ with discrete spectrum, we have
\beq\label{eq:InvProb1}
\| {H}^{-1} \| > \frac{1}{t} \iff  \sigma(H) \cap \left( - \frac{1}{t}, \frac{1}{t} \right) \neq \emptyset ,
\eeq
so that
\beq\label{eq:InvProb2}
\| {H}^{-2} \| > \frac{1}{t^2} \iff  \sigma(H^2) \cap \left( 0, \frac{1}{t^2} \right) \neq \emptyset .
\eeq
This, in turn, is equivalent to the fact that $H^{2} + \frac{1}{t^2}$ has an eigenvalue in $\left( \frac{1}{t^2}, \frac{2}{t^2} \right)$.
So for a random self-adjoint matrix $H$, we have
\bea\label{eq:fluxEst1}
\P \{ \| H^{-1} \| > t \} & = & \P \left\{ \sigma \left( H^2 + \frac{1}{t^2} \right) \cap \left( \frac{1}{t^2}, \frac{2}{t^2}          \right) \neq \emptyset                  \right\} \nonumber \\
             & \leq & \P \left\{ \| (H^2 + {t^{-2}} )^{-1} \|  > \frac{t^2}{2} \right\} \nonumber \\
             & \leq & \P \left\{ \| (H^2 + t^{-2} )^{-1} \|_2  > \frac{t^2}{2} \right\} \nonumber \\
              & \leq & \frac{2}{t^2} \E \left\{ {\rm Tr} \left[ \left (H^2 + \frac{1}{t^2} \right)^{-1} \right] \right\} .
\eea
The final simplification comes from the fact that
\beq\label{eq:resolventID1}
(H^2 + t^{-2} )^{-1} = - t \Im ( H + i t^{-1})^{-1},
\eeq
so substituting the right side of \eqref{eq:resolventID1} into the last line of \eqref{eq:fluxEst1},
we obtain
\beq\label{eq:InvProb3}
\P \{ \| H^{-1} \| > t \} = - \frac{2}{t} \sum_{i=1}^2 \E \left\{ \langle e_i, \Im (H + i t^{-1} )^{-1} e_i \rangle \right\} .
\eeq
To evaluate the expectation on the right of \eqref{eq:InvProb3}, we use the Schur complement formula of Lemma \ref{lemma:schur1} for the rank-one projection $P = P_i$, the projection onto $e_i$. Letting $Q := 1 - P_i$, we define $\Gamma_i := \langle e_i, HQ ( QHQ + it^{-1}Q )^{-1} QH e_i \rangle$, for $i = 1,2$. The scalar $\Gamma_i$ is independent of the random variable $v_i$.  We then have
\bea\label{eq:schur4}
\langle e_i, \Im (H + i t^{-1} )^{-1} e_i \rangle  & = & \Im (v_i + i t^{-1} + \Gamma_i )^{-1} \nonumber \\
  & = & - (t^{-1} + \Im \Gamma_i) [ (v_i + \Re \Gamma_i)^2 + ( t^{-1} + \Im \Gamma_i)^2]^{-1} . \nonumber \\
   &
  \eea
We now integrate the right side of \eqref{eq:schur4} with respect to $v_i$ and obtain the upper bound $\pi \| \rho\|_\infty$.
Returning to \eqref{eq:fluxEst1}, we obtain
\beq\label{eq:final1}
\P \{ \| (V + A)^{-1} \| > t \} \leq \frac{4 \pi \| \rho \|_\infty}{t},
\eeq
proving the lemma.
 \end{proof}

Given this technical lemma, we can prove Proposition \ref{proposition:SpectralAve1}.

\begin{proof}
\noindent
1.  We first prove \eqref{eq:matrixEle1} on the diagonal $j = k$. Let $H_{(\omega_j^\perp, 0)}^N$ be the matrix $H_\omega^N$ with the $jj$ entry set to $0$.  Using the second resolvent identity, we have the rank one perturbation formula
\beq\label{eq:SAdiagonal1}
\langle e_j, \left(H_\omega^N - z\right)^{-1} e_j \rangle
= \left( \omega_j + \langle e_j, (H_{(\omega^\perp_j,0)}^N -z)^{-1} e_j \rangle^{-1} \right)^{-1}.
\eeq
Thus, we have
\beq\label{eq:SAdiagonal2}
\E  \{ | \langle e_j,  (H_\omega^N -z )^{-1}  e_j \rangle |^s \}
= \E_{\omega_j^\perp}  \left\{  \int_\R  d\omega_j   \rho(\omega_j)  \left| \left( \omega_j +
\langle e_j, (H_{(\omega_j^\perp,0)} -z )^{-1}    e_j \rangle^{-1}  \right)^{-1} \right|^s  \right\} .
\eeq
Since $s \in (0,1)$ and $\rho \geq 0$ is a probability density, it follows that
for any $a \in \C$,
\beq\label{eq:SAdiagonal3}
 \int_\R d\omega_j \rho(\omega_j) ~|{\omega_j +a }|^{-s} \leq C_{\rho,s} < \infty,
 \eeq
 independent of $a$. Taking $a = \langle e_j, (H_{(\omega_j^\perp,0)}-z )^{-1}    e_j \rangle^{-1}$, it follows that the expectation in \eqref{eq:SAdiagonal2} is uniformly bounded in $z \in \C$.

\noindent
2. For the off-diagonal terms $j \neq k$ in \eqref{eq:matrixEle1},
we let $P_{jk}$ be the orthogonal projection onto the span of $e_j$ and $e_k$. We begin with energies $\lambda \in \R$. Then,
we have
\beq\label{eq:SAoffdiagonal1}
| \langle e_j, (H_\omega^N- \lambda )^{-1} e_k \rangle | \leq \| P_{jk} (H_\omega^N- \lambda )^{-1} P_{jk} \|   ,
\eeq
where the right hand side is the operator norm of the $2 \times 2$ matrix.  From the Schur complement formula, Lemma \ref{lemma:schur1}, we can write the $2 \times 2$ as $P_{jk} (H_\omega^N- \lambda )^{-1} P_{jk}$ as
\beq\label{eq:OffDiag1}
P_{jk} (H_\omega^N - \lambda )^{-1} P_{jk} =  (V_{jk} + A(\lambda)  )^{-1},
\eeq
where the diagonal potential $V_{jk}$ is
\beq
V_{jk}= \left(\begin{array}{cc}
\omega_{j} & 0 \\
0 & \omega_{k} \\
\end{array}\right)
\eeq
and $A(\lambda)$ is a $2 \times 2$ independent of $V_{jk}$ that is self-adjoint for $\lambda \in \R$, which we now assume.
Using the layer cake representation of the expectation, we have
\beq\label{eq:OffDiag2}
 \E \{ \| \left(V_{jk} - A(\lambda) \right)^{-1} \|^s  \}
= \int_0^\infty \P  \left\{ \| \left(V_{jk} - A(\lambda) \right)^{-1} \|^s >  t\right\} \, dt.
\eeq
We now apply Lemma \ref{lemma:SpectralAve2} to the integrand on the right in \eqref{eq:OffDiag2}. For $t > 1$, the integral is finite since the probability behaves like $t^{-\frac{1}{s}}$ and $s \in (0,1)$. For $t \in [0,1]$, we use the fact that the probability is bounded by one. Bounds \eqref{eq:OffDiag2} and \eqref{eq:SAoffdiagonal1} establish \eqref{eq:matrixEle1} for $\lambda \in \R$, with constant independent of $\lambda$.

\noindent
3. To extend the bound \eqref{eq:matrixEle1} from $\lambda \in \R$ (proven in part 2 above) to $z \in \C$, we follow an idea in \cite[Theorem B.1]{asfh}. For $z \in \C^+$, the function
\beq\label{eq:subharm1}
f_{jk}(z) :=  \E  \{ | \langle e_j,  (H_\omega^N -z )^{-1}  e_k \rangle |^s \}
\eeq
is subharmonic with boundary-values on $\R$ that exist almost everywhere. Furthermore, by part 2 of the proof, the function  $f_{jk}(\lambda)$ is uniformly bounded for $\lambda \in \R$. Thus, by the Poisson representation
of a subharmonic function on $\C^+$, we have
\beq\label{eq:subharm2}
f_{jk}(z) \leq \frac{1}{\pi} \int_{\R} ~f_{jk}(\lambda) \frac{y}{(x-\lambda)^2 + y^2} ~d \lambda \leq C_{\rho,s}.
\eeq
This proves the bound for all $z \in \C^+$. Using the Poisson representation for $\C^-$ establishes the result for $z \in \C$.

\noindent
4. We now turn to the \textit{a priori} estimate \eqref{eq:ImMatrixEle1}.
Let $z= E + i \epsilon$, for $\epsilon > 0$.   Using the Schur complement formula in Lemma \ref{lemma:schur1} with $P = P_j$, the rank one projection onto $e_j$, we obtain
\beq\label{eq:schurRank1}
\E  \left\{  \Im\left\langle e_j, \left(H_\omega^N - E-i\eps \right)^{-1} e_j \right\rangle \right\}
=\E\left\{ \Im  ( \omega_j -E- i\eps + a(z) )^{-1 } \right\} ,
\eeq
where $a(z)$ is independent of $\omega_j$.
Writing  $\tilde{E}=E + \Re a$ and $\tilde{\epsilon}= \epsilon + \Im a$. we obtain from \eqref{eq:schurRank1},
\beq\label{eq:schurRank2}
\E\left\{\Im (\omega_j -E-i\epsilon + a )^{-1 } \right\}
=\E\left\{    \frac{\tilde{\epsilon}}{(\omega_j-\tilde{a})^2 + \tilde{\eps}^2}        \right\} .
\eeq
The result now follows by integration of \eqref{eq:schurRank2} with respect to $\omega_j$. 
This  bound is uniform in $E$ and independent of $\eps > 0$. Replacing the imaginary part by the absolute value in the left side of \eqref{eq:schurRank1}, we obtain the corresponding bound for $\epsilon < 0$. 
\end{proof}

\subsection{Wegner and Minami estimates}\label{subsec:WMbounds1}

Localization and LES require two eigenvalue estimates:  a Wegner estimate and a Minami estimate. These are easily obtained by averaging over the diagonal terms only. Since the random variables $v_{ij}$ are scaled with the bandwidth like $L^{- \frac{1}{2}}$, the upper bound in the Wegner estimate scales as $L^{\frac{1}{2}}$. As mentioned above, the more refined calculation of \cite{psss} results in an upper bound independent of the bandwidth for Gaussian random variables. This is anticipated to be an  important improvement for the cases for which the bandwidth grows with $N$. 

We now return to the fixed-width RBM $H_L^N$ as in \eqref{eq:RBMdefn1}.
We begin with the Schur complement formula in section  \ref{subsec:schur1} with $P=P_j$, the rank-one projection onto the subspace of $\C^{2N+1}$ generated by $e_j$, for $j \in \{ -N, \ldots, N\}$, and we write $Q : = 1 - P_j$ for the orthogonal projection. We refer to the proof of Lemma \ref{lemma:SpectralAve2}, near \eqref{eq:schur4}, for a similar calculation and the definition of $\Gamma_j$.  Using this formula, we obtain the spectral averaging bound
\bea\label{eq:schur1}
 \Im \langle e_j, (H_L^N - E - i \eps )^{-1} e_j \rangle  & \leq & 
\int ~ \rho(v_{jj}) \Im  ( v_{jj} - E - i \eps +   \Gamma_j(E+i \eps)  )^{-1} ~d v_{jj} \nonumber \\
& \leq & L^{\frac{1}{2}} \pi \| \rho \|_\infty,
\eea
uniformly in $E$ and in $\eps > 0$.
This estimate and the method of Combes, Germinet, and Klein \cite{cgk1} leads to the following bounds.
For a self-adjoint operator $A$, we write $P_I(A)$ for the spectral projection associated with $A$ and the interval $I \subset \R$. 

\begin{prop}[Wegner and Minami estimates] \label{prop:evCorrelations1}
For any $(2N+1) \times (2N+1)$ real symmetric matrix $H_L^N$ with band width $2L+1$ and with diagonal elements $\{ v_{jj} \}$ $iid$ with density $\rho$, we have
\beq\label{eq:wegnerEst1}
\P \{ {\rm Tr} P_I(H_L^N ) \geq 1 \} \leq \E  \{ {\rm Tr} P_I({H_L^N}) \} \leq  \pi \| \rho \|_\infty L^{\frac{1}{2}} (2N+1) |I|.
\eeq
and
\bea\label{eq:minamiEst1}
\P \{ {\rm Tr} P_I({H_L^N}) \geq 2 \}  & \leq  & \E  \{ {\rm Tr} P_I({H_L^N}) ( {\rm Tr} P_I({H_L^N}) - 1)  \} \nonumber \\
  &  \leq  & (  \pi \| \rho \|_\infty L^{\frac{1}{2}} (2N+1) |I| )^2.
\eea
\end{prop}

Although not needed in the present work, we mention the following eigenvalue correlation estimate that is a version of Theorem 3 of Peled, Schenker, Shamis, and Sodin \cite{psss} for the case of RBM $H_L^N$, generalizing the Minami estimate, mentioned above.

 \begin{theorem} \cite[Theorem 3]{psss} \label{thm:psss1}
 Let $H_L^N$ be a symmetric random band matrix with diagonal random variables $iid$ with an absolutely continuous probability measure with density $\rho$. Then, we have
 \beq\label{eq:genMinami1}
\E \left\{ \prod_{\ell=0}^{m-1}\left( \tr P_I({H_L^N})-\ell  \right) \right\}
 \leq \left( C L^{1/2}(2N+1)|I|\right)^m,
\eeq
 where the constant $C$ depends on $\rho$. In the case when $\rho$ is Gaussian and the diagonal terms $v_{jj}$ are random variables with variance one, there is no factor of $L^{\frac{1}{2}}$ in \eqref{eq:genMinami1}, that is, the upper bound is independent of the bandwidth.
\end{theorem}

The Wegner estimate corresponds to $m=1$ and the Minami estimate to $m=2$ in \eqref{eq:genMinami1}.


\subsection{Localization bounds}\label{subsec:localBounds1}

Schenker proves localization estimates on the $s$-moment of the matrix elements of the resolvent at real energies in any interval $[-r, r] \subset \R$ with constants depending on $r > 0$.

\begin{theorem}\cite{schenker}\label{thm:fractMoment1}
Given $r>0$ and $s\in(0,1)$, there are constants $\mu_{r,s}>0$, $C_{r,s}<\infty$, and $\alpha_{r,s}>0$ such that
\beq\label{eq:locBound1}
 \E\left\{ \left| \left\langle e_j, \left(H_L^N- E \right)^{-1} e_k
\right\rangle \right|^s \right\} \leq C_{r,s} L^{s/2} e^{-\alpha_{r,s} {L^{- \mu_{r,s}}}  {|j-k|} }
\eeq 
for all $E \in [-r,r]$, and all $i,j=-N, \dots, N$.
\end{theorem}

We need this type of estimate for all $z \in \C^+$. We will prove it by matching the Schenker estimate \eqref{eq:locBound1}
with the following high-energy estimate.

\begin{theorem}[Localization at Extreme Energies]\label{thm:locHE1}
For any $s \in (0,1)$ and for some fixed value $R>0$, there exist constants $C_{R,s}>0$ and $\alpha_{R, s}>0$ such that for each $z\in \C$ with $|z|>R$,
\beq\label{eq:locBoundHE1}
 \E\left\{ \left| \left\langle e_j, \left(H_L^N-z\right)^{-1} e_k
\right\rangle \right|^s \right\} \leq C_{R,s}  e^{-\alpha_{R,s} |j-k|}.
\eeq
\end{theorem}

 The proof follows the iteration procedure of Aizenman-Molchanov [AM] and is carried out in the appendix, section \ref{sec:loc1}. Combining Theorem \ref{thm:fractMoment1} and Theorem \ref{thm:locHE1}, we prove exponential decay of the $s$-moment of the matrix elements of the resolvent along the real axis and, by a subharmonicity argument, in $\C$.

\begin{theorem}[Localization at all energies]\label{thm:locAllE1}
For any  $s \in (0,1)$, and for all $z\in \C$, there exist finite constants $C_{L,s} > 0$ and $\alpha_{L,s} >0$, depending on $L$ and $s$ (but uniform in $z \in \C$), such that
\beq\label{eq:ExpDecay1}
 \E\left\{ \left| \left\langle e_j, \left(H_L^N-z\right)^{-1} e_k
\right\rangle \right|^s \right\} \leq C_{L,s} e^{-\alpha_{L,s} |j-k|}.
\eeq
\end{theorem}

\begin{proof}
We first obtain uniform bounds on $\R$. 
For any  $s \in (0,1)$, there exists $R >0$ as in Theorem \ref{thm:locHE1} so that the fractional moment bound \eqref{eq:ExpDecay1} holds for $|E| > R$.  We then choose $r > R > 0$ (depending on  $s \in (0,1)$) as in Theorem \ref{thm:fractMoment1}. We set $C_{L,s} := \max \{ C_{R, s}, C_{r, s} \}$ and $\alpha_{L,s} := \min \{ \alpha_{R,s}, \alpha_{r, s} \} > 0$. It follows from Theorem \ref{thm:fractMoment1} and Theorem \ref{thm:locHE1}, that for these constants, uniformly in $E \in \R$, we have
\beq\label{eq:ExpDecayReal1}
\E\left\{\left| \left\langle e_j, \left(H_L^N - E \right)^{-1} e_k\right\rangle  \right|^s\right\}
\leq C_{L,s} e^{-\alpha_{L,s} |j-k|}.
\eeq
We now extend \eqref{eq:ExpDecayReal1} to $\C^+$ using subharmonicity. We mention that similar arguments are given in \cite{asfh} and in \cite{nakano1}.
We note that the function
\beq\label{eq:subharm3}
 f_{jk}(z): = \E\left\{\left| \left\langle e_j, \left(H_L^N - z \right)^{-1} e_k\right\rangle \right|^s\right\} \geq 0,
\eeq
is subharmonic in the upper half-plane with boundary values that exist everywhere on $\R$ by the spectral averaging estimate \eqref{eq:matrixEle1}.
Thus, setting $z=x+iy$ with $y>0$ and using the Poisson kernel representation for harmonic functions in the upper half-plane, we have
\beq\label{eq:subharm4}
 f_{jk}(z)\leq \frac{1}{\pi}\int_\R f_{jk}(E) \frac{y}{(x-E)^2+y^2}\, dE.
\eeq
Using the localization bound which is now uniform in $E \in \R$, we obtain from \eqref{eq:ExpDecayReal1}- \eqref{eq:subharm4}
\beq
 f_{jk}(z) \leq \frac{1}{\pi} C_{L,s} e^{-\alpha_{L,s} |j-k|}   \int_\R ~ \frac{y}{(x-E)^2+y^2} = C_{L,s} e^{-\alpha_{L,s}|j-k|}.
\eeq
The estimate also holds for $y<0$ using the Poisson kernel for the lower half-plane.
\end{proof}


\section{Density of States}\label{sec:dos1}
\setcounter{equation}{0}

In this section, we define and review some properties of the local integrated density of states ($\ell$IDS), and the corresponding local DOS measure ($\ell$DOSm), and their $N \rightarrow \infty$ counterparts.   
As in from section \ref{subsec:mainResult1}, we let   $\{E_L^N(j)\}_{j=-N}^{N}$ denote the set of eigenvalues of the random matrix $H_L^N$. The {\it local integrated density of states} ($\ell$IDS) is the expectation of the normalized eigenvalue counting function:
\beq\label{eq:lids1}
 N_L^N(E) = :  \frac{1}{2N+1} \, \E \{  \#\{ j \, : \, E_L^N(j)\leq E \} \} .
\eeq
This is the distribution function of a point measure on $\R$ called the local DOS measure ($\ell$DOSm) denoted by $\nu_L^N$:
$$
N_L^N(E) :=  \int_{- \infty}^E ~ d \nu_L^N (s) .
$$ 
The integrated density of states (IDS) is the limit of the $\ell$IDS as $N \rightarrow \infty$:
\beq\label{eq:ids11}
 N_L ^\infty(E) = : \lim_{N \rightarrow \infty}  \frac{1}{2N+1} \, \E \{  \#\{ j \, : \, E_L^N(j)\leq E \} \} .
\eeq
and is the distribution function of the density of states measure (DOSm) $\nu_L$. 
When it exists, the derivative of the $\ell$IDS is the {\it local density of states function} ($\ell$DOSf):
\beq\label{eq:localDosf1}
n_L ^N(E) =: \frac{d}{dE} N_L^N (E).
\eeq
Similarly, the derivative of the IDS is the
 {\it density of states function} (DOSf):
\beq\label{eq:dosf1}
n_L^\infty (E) =: \frac{d}{dE} N_L^\infty  (E).
\eeq

Under Assumption 1 on the random variables, the DOSm and IDS exist. In section \ref{sec:ptwConvDOS1} we will prove that for $k \geq 2$ in Assumption 1, the DOSf exists and is $C^k$.


\subsection{Smoothness of the local IDS}\label{subsec:smoothIDS1}

In this section, we prove that because of the expectation in \eqref{eq:lids1} the $\ell$IDS is as smooth as the probablity density $\rho$. We present this short proof, using some ideas from \cite{dkm}, as the technique will be in the proof of Theorem \ref{thm:smoothDos1}.

\begin{prop}\label{prop:smoothLIDS1}
Let $H_L^N$ be an $iid$ random band matrix with probability density function $\rho$ satisfying Assumption 1 with support in $I_\rho$. For example, we make take $\rho$ to be Gaussian.  Then the local integrated density of states $N_L^N (E)$, defined in \eqref{eq:lids1}, is in $C^k(\R)$.
\end{prop}

\begin{proof}
As in section \ref{subsec:WMbounds1}, we let ${P}_{(-\infty, E]}(H_L^N)$ be the spectral projection for $H_L^N$ and the interval $(-\infty, E].$  By a simple change of the spectral parameter, we may write the $\ell$IDS as
\bea\label{eq:changeVariable1}
N_L^N(E) & = & \frac{1}{2N+1}\E\left\{\tr{P}_{(-\infty, E]}(H_L^N)\right\} \nonumber \\
 & = & \frac{1}{2N+1} \E \left\{ \tr {P}_{(-\infty,0]}(H_L^N- E I_N)\right\}  ,
\eea
where $I_N$ is the identity matrix. We define new diagonal random variables by $\tilde{v}_{ii} := v_{ii} - E$, for $i = - N , \ldots, N$. These new random variables have a density $\rho (\tilde{v}_{ii} + \lambda)$.
In terms of these new random variables, the $m^{th}$ derivative, for $1 \leq m \leq k$, of the $\ell$IDS may be written as
\bea\label{eq:changeVariable2}
 \left(\frac{d^m  N_L^N}{d E^m }\right) (E)  & = &
 \frac{1}{2N+1} \frac{d^m}{d E^m} \left( \prod_{i=-N}^N   \int_{I_\rho- E} d \tilde{v}_{ii} ~ \rho ( {\tilde{v}}_{ii} + E) \right)  \nonumber \\
 & &  \times \prod_{\substack{ i<j: \\ |i-j|\leq L}}\int_{I_M} ~ dv_{ij} \rho (v_{ij})
 ~\tr {P}_{(-\infty, 0]} \left(H_L^N \right).
\eea
The integral over the diagonal random variables $\tilde{v}_{ii}$ is the convolution of a Schwartz function with the trace averaged over the off-diagonal random variables $v_{ij}$, $i \neq j$. As this is an $L^1$ function of $\{ \tilde{v}_{ii}, ~i = -N, \ldots, N \}$, we can bring the derivative onto the product of the probability densities of the diagonal random variables $\rho(\tilde{v}_{ii}+ E)$. The boundary terms vanish as $\rho \in C^k_\infty (\R)$. 
\end{proof}

Unfortunately, the above proof of existence does not give us uniform estimates on the derivatives in either the dimension $N$ or band width $L$.  We do have the following corollary of the Wegner Estimate (Theorem 2 with $m=1$). \\

\begin{corollary}\label{corollary:ldosf1}
Let $H_L^N$ be as in Proposition \ref{prop:smoothLIDS1}.  Then the local density of states function $n_L^N (E) \in C^{k-1} (\R)$ if $\rho \in C^k(\R)$. Furthermore, the $\ell$DOS is  uniformly bounded in $N$.
\end{corollary}

\begin{proof}
Smoothness in $E$ follows from Proposition  \ref{prop:smoothLIDS1}. As for the uniform bound in $N$, we note that
for any energy $E$ and integer $N$,
\bea\label{eq:Nbdd1}
n_L^N(E)   & =    & \lim_{h\to 0} \frac{1}{h} [ N_L^N( E +h) - N_L^N( E) ]  \nonumber \\
   &  = &  \left( \frac{1}{2N+1} \right) ~  \lim_{h\to 0} \frac{1}{h} \E\{  P_{[E,E+h]} (H_L^N) \}     \nonumber \\
    & \leq  & 2 \pi \| \rho \|_\infty L^{\frac{1}{2}} ,
\eea
where the Wegner estimate,  Proposition \ref{prop:evCorrelations1},  was used to obtain the last inequality.
\end{proof}

\noindent
If the random variables are Gaussian, it was proved in \cite{psss} that the factor of $L^{1/2}$  in \eqref{eq:Nbdd1} may be replaced by one, see Theorem \ref{thm:psss1}. 

\vspace{.1in}

We will use the following representation for the density of states function:

\begin{prop}\label{prop:dosDefn1}
For all $E \in \R$, the $\ell$DOS function has the representation
\beq\label{eq:dosDefn1}
 n_L^N(E) = \frac{1}{2N+1} \frac{1}{\pi} \sum_{j=-N}^N  \lim_{\eps\rightarrow 0} \E  \{\Im \langle e_j, (H_L^N- E - i \eps )^{-1} e_j \rangle  \}.
\eeq
\end{prop}

\begin{proof}
From Stone's Formula for spectral projectors, the $\ell$IDS
may be expressed as
\[ N_L^N(E+h)-N_L^N(E)
=\frac{1}{\pi} \frac{1}{2N+1} \sum_{j = -N}^N  \lim_{\eps \rightarrow 0} \int_{E}^{E+h}  \E \{ \Im \langle e_j, (H_L^N- \lambda - i \eps )^{-1} e_j \rangle \} \, d\lambda.
\]
By the Dominated Convergence Theorem, we can bring the $\eps$ limit inside of the $\lambda$ integral so that 
\bea\label{eq:ldos3}
n_L^N(E)  &=  & \lim_{h\to 0}\frac{1}{h} [ N_L^N(E+h)-N_L^N(E) ] \nonumber \\
   & = & \lim_{h\to 0} \frac{1}{\pi}  \frac{1}{2N+1} \sum_j \frac{1}{h} \int_{E}^{E+h} \lim_{\eps \rightarrow 0} \E \{ \Im \langle e_j, (H_L^N- \lambda - i \eps )^{-1} e_j \rangle \} \, d\lambda.
\eea
The representation then follows by the Lebesgue Differentiation Theorem that can be applied because of the regularity of the $\ell$DOSf stated in Proposition \ref{prop:smoothLIDS1}.
\end{proof}



\subsection{Infinite-volume density of states}\label{subsec:InfDOS1}

Since the band width $L$ is fixed as $N$ increases, the matrices $H_L^N$ have a natural limiting operator $H_L^\infty$ on $\ell^2(\Z)$, which is defined by its matrix elements in the same way as $H_L^N$. This real symmetric operator $H_L^\infty$
has matrix elements given by
\beq\label{eq:infDos0}
\langle e_i,  H^\infty_L e_j \rangle =\frac{1}{\sqrt{2L+1}}\left\{\begin{array}{ccc}
v_{ij} & if & |i-j| \leq L \\
0  & if & |i-j| > L\end{array}\right.
\eeq
with $\langle e_i,  H^\infty_L e_j \rangle = \langle e_j,  H^\infty_L e_i \rangle$, for
$-\infty \leq i,j \leq \infty.$
In this case, the local matrix $H_L^N= \chi_{[-N,N]} H_L^\infty \chi_{[-N,N]}$ is the restriction of $H_L^\infty$ to a box with simple boundary conditions.

The {\it infinite volume density of states function} $n_L^\infty(E)$ is defined as follows.
Since the infinite-volume operator $H_L^\infty$ is ergodic under translation in $\Z$, the Birkhoff Ergodic Theorem
gives
\beq\label{eq:infDos1}
\lim_{N\to\infty}
\frac{1}{2N+1}\sum_{j=-N}^N  \E \{ \Im \langle e_j, \left(H_L^\infty - z \right)^{-1} e_j \rangle \} =
\E \{ \Im  \langle e_0, \left(H_L^\infty - z \right)^{-1}e_0 \rangle \}  .
\eeq
So for almost every energy $E$, we define the DOSf as
\beq\label{eq:infDos2}
n_L^\infty(E) := \lim_{\eps \rightarrow 0} \frac{1}{\pi}  \E \{ \Im  \langle e_0, \left(H_L^\infty - E - i \eps \right)^{-1}e_0 \rangle \}  .
\eeq
In section \ref{subsec:unifConvDOS1}, we will prove the uniform convergence of the $\ell$DOSf to the DOSf defined in \eqref{eq:infDos2}.

\begin{remark}\label{remark:SclIds1}
The structure of the IDS $N_L^\infty (E)$ was studied by Bogachev, Molchanov, and Pastur \cite{bmp} for $d=1$ RBM. These authors characterized the IDS for different growth rates of the bandwidth $b_N$: 1)  $b_N \rightarrow \infty$ as $N \rightarrow \infty$, and 2) $b_N =  b_0 \geq 1$, a constant.    They proved the following theorem that applies to fixed bandwidth RBM \cite[section 6]{bmp}.

\begin{theorem}\cite[Theorem 5]{bmp}\label{thm:BMPids1}
For a real symmetric fixed bandwidth RBM $H_L^N$ under Assumption 1 and the additional condition that the density $\rho$ be even, with $\E \{ v_{ij} \} = 0$, the moments of the 
$\ell$DOSm $\nu_L^N$, defined by
\beq\label{eq:moments1}
\mu_{N,L}^{(p)}   := \int E^p d \nu_L^N (E)  ,
\eeq
converge as $N \rightarrow \infty$ to the moments of  a measure $\nu_L$. The odd moments vanish and the even moments satisfy
$$
\mu_{L}^{2p} : =  \int E^{2p} ~ d \nu_L (E) = \mu_{SCL}^{2p} + \mathcal{O}(L^{-1}) .
$$
where $\nu_{SCL}$ is the measure associated with the semicircle law (SCL):
$$
d\nu_{SCL} (E) = \frac{1}{2 \pi} ( 4 - E^2)_+^{\frac{1}{2}}~dE  .
$$
Consequently, the DOSm $\nu_L$ converges  to $\nu_{SCL}$ as $L \rightarrow \infty$ in the sense of moment convergence as $L \rightarrow \infty$.
Furthermore, if the support of the random variables is noncompact, then the support of $\nu_L$ is unbounded. 
\end{theorem}

The theorem states that the DOSm $\nu_L$ converges  to $\nu_{SCL}$ as $L \rightarrow \infty$ in the sense of moment convergence so the tails of the support of the measure $\nu_L$ vanish in this limit.
For a Gaussian probability density $\rho$, this theorem indicates that the support of the DOSm $\nu_L$ is the real line and resembles the profile of the semicircle law near $[-2,2]$ with long tails outside of this set that vanish as $L \rightarrow \infty$.

\end{remark}

\section{A fractional moment bound for fixed-width random band matrices}\label{sec:fracMoment1}
\setcounter{equation}{0}

As above, we let $H_L^N$ denote a $(2N+1) \times (2N+1)$ real symmetric matrix (an extension to hermitian matrices follows the same arguments) as described in \eqref{eq:RBMdefn1}. We label the rows and columns $(i,j)$, with $-N \leq i,j \leq N$, and the upper left entry has indices $(-N, -N)$. 
 By the band width of $H_L^N$, we mean that there are $(2L + 1)$ non-trivial entries in any row or column of $H_L^N$. We always consider $0 \leq L \ll N$. For example, in the $i^{\rm th}$-row, the column index $j$ for nontrivial entries lies in the range:
\beq\label{eq:columnIndex1}
\max \{-N, i - L \} \leq j \leq \min \{ i+L, N \}, 
\eeq
 and similarly for columns, and $[H_L^N]_{ij} = 0 $ if $|i - j| > L$.
As in section \ref{sec:fundamentals1}, for many arguments, it is sufficient to assume that only the diagonal matrix elements $[H_L^N]_{ii} = v_{ii}$ are random variables as in Assumption 1. 
%
The off-diagonal entries of $H_L^N$ need not be random so that fixed-width real symmetric matrices with diagonal randomness are included in the class of matrices for which our results hold.  These models include the Anderson model on the $d$-dimensional lattice and its higher-rank generalizations (see \cite{dkm}). 


For $N >  M \gg L \geq 0$, we need to compare compare $(2N+1) \times (2 N +1)$ and $(2M+1) \times (2 M +1)$ real symmetric matrices.
We write $\C^{2N+1}$ for the $(2N+1)$ dimensional vector space with standard orthonormal basis $e_j$ for $-N \leq j \leq N$. For $N > M$, we naturally embed $\C^{2M+1}$ into $\C^{2N+1}$ using this basis $e_j$ for $-M \leq j \leq M$. Similarly, we embed $(2M+1) \times (2M+1)$ matrices into $(2N+1) \times (2N +1)$ matrices using this labeling of the basis vectors. In this way, the zeroth row and zeroth column are preserved for $N > M \gg L \geq 0$. Finally, we denote by $A^{-1}$ the inverse of the matrix $A$ on its range and, if considered embedded in a large space, all other matrix elements are set equal to zero. 


In the following theorem,  the probability density $\rho$ may have noncompact support,
removing the restriction in \cite{dkm}. The following theorem generalizes and simplifies \cite[Theorem 2.2]{dkm}.

\begin{theorem}\label{thm:fractionalMoment1}
Suppose that the probability density $\rho$ satisfies Assumption 1.
For any $0<s<1$, integers $N > M \gg L \geq 0$, $\epsilon > 0$, and $E \in \R$, there exists a finite constant $C_{s,\rho} > 0$, 
so that for $-M +L \leq j \leq M - L$,
\bea\label{eq:fracMoment2}
\lefteqn{ \left|\E \left\{\left\langle e_j, \left(H_L^N -E - i \eps\right)^{-1} e_j \right\rangle - \left\langle e_j, \left(H_L^M -E - i \eps\right)^{-1} e_j \right\rangle \right\} \right| } \nonumber \\
 & \leq  & C_{s, \rho} \E \left\{ \left| \left\langle \Psi_j, \left[\left(\widetilde{H}_L^N(j) -E-i\eps\right)^{-1} - \left(\widetilde{H}_L^M(j)-E-i\eps\right)^{-1}\right]  \Psi_j
\right\rangle\right|^s \right\},
\eea
where $\widetilde{H}_L^N(j)$ and $\widetilde{H}_L^M(j)$ are $(2N+1) \times (2N+1)$, respectively, $(2M+1) \times (2M+1)$ matrices, constructed from $H_L^N$ and $H_L^M$, respectively, by setting all the entries in the $j^{\rm th}$-row and $j^{\rm th}$-column equal to zero.
The random vector $\Psi_j \in \R^{2N+1}$ has non-zero entries occurring only for indices between $j-L$ and $j+L$ and with the $j^{\rm th}$-entry equal to zero. 
\end{theorem}

\begin{proof}
1. First we use the the Schur complement formula presented in Lemma \ref{lemma:schur1}  with $P=|e_j\rangle \langle e_j|$, the rank-one  projection onto the vector $e_j$, in order to obtain an expression for
$P(H_L^N -z)^{-1}P$, where $z = E + i \eps$,  in which the random variable $v_{jj}$ is explicit. We have
\bea\label{eq:schur2}
{ P (H_L^N -z )^{-1} P } &=& [  P (H_L^N-z) P  \nonumber  \\
 & &   - P H_L^N Q  \left(Q \left(H_L^N-z\right) Q \right)^{-1} Q H_L^N P ]^{-1} .
\eea
To analyze this formula, we note that
\beq\label{diagonal1}
P \left(H_L^N \right) P = \left[H_L^N\right]_{jj} P =: v_{jj} (N)P ,
\eeq
and that the matrix $Q H_L^N P$ is the matrix whose $j^{\rm th}$ column is the $j^{\rm th}$ column of $H_L^N$ with the ${jj}^{\rm th}$-entry set equal to zero. That is, the only nonzero entries of $Q H_L^N P$ are in the $j^{\rm th}$ column with indices $i$ satisfying $j-L \leq i \leq j+L$ and $i \neq j$. 
In particular,
 this matrix $Q H_L^N P$ is independent of $v_{jj}(N)$. This matrix has the form
\beq\label{eq:randomVec1}
 Q H_L^N  P = \sum_{\substack{i= j-L,\dots,j+ L \\ i \neq j}}
\left[H_L^N\right]_{ij}  e_i  \otimes e_j ,
\eeq
where $\otimes$ is the outer product of vectors. 
We can interpret this matrix as a column vector which we denote by
$$
\widetilde{ \Psi}_j(N ) : = Q  H_L^N P e_j ,
$$
that is also independent of $v_{jj}(N)$.
 The operator $Q H_L^N Q$  occurring in the Schur complement formula \eqref{eq:schur2}  will be denoted by $\widetilde{H}_L^N(j)$. This $(2N +1) \times (2N + 1)$ matrix is obtained from $H_L^N$ by setting the entries in the $j^{\rm th}$-row and the $j^{\rm th}$-column equal to zero. 
Since $P = P_j$ is the rank-one projection onto $e_j$, we may write \eqref{eq:schur2} as
\beq\label{eq:schur3}
\left\langle e_j, \left(H_L^N -E - i \eps\right)^{-1} e_j \right\rangle
= \left[ v_{jj} - E - i \eps - \Psi_j^T \left(\widetilde{H}_L^N(j) - E-i\eps \right)^{-1} \Psi_j \right] ^{-1}.
\eeq
where $\Psi_j^T$ is the transpose of the column vector $\Psi_j$.


\noindent
2. We now consider integers $0 \leq L \ll M < N$, and the two matrices $H_L^N$ and $H_L^M$ coming from the same sample of random variables. We then have $v_{jj}(M)= v_{jj}(N) := v_{jj}$. Furthermore, we can identify $\widetilde{\Psi}_j(M) \in \C^M$ with a vector $\widetilde{\Psi}_j (N) \in \C^N$ by setting the extra entries equal to zero. With this identification, we have $\widetilde{\Psi}_j(M)=\widetilde{\Psi}_j(N)$. We denote this $(2N+1) \times 1$ column vector by $\Psi_j$ and note that it has $2L$ nonzero entries with indices between $j-L$ and $j+L$.  
Furthermore, since the width of the matrices is a constant, the matrix $\widetilde{H}_L^N(j)$ differs from $\widetilde{H}_L^M(j)$ (recall $N > M$) only in the entries $(i,j)$ for which $|i| > M$ or $|j| > M$ (see \eqref{eq:upperCap1}--\eqref{eq:lowerCap1}). 
We write $\widetilde{R}_{N,L,j} (z)$ for the resolvent $(\widetilde{H}_L^N(j) - z)^{-1}$. 
Letting $z = E + i \eps$ and recalling $|j| \leq M-L$, we derive a representation for the 
difference
\beq\label{eq:differenceSchur1}
\E \{ \left\langle e_j, {R}_{N,L}(z)  e_j \right\rangle - \left\langle e_j,  R_{M,L}(z)  e_j \right\rangle \}.
\eeq
The denominator on the right side of \eqref{eq:schur3} has a negative imaginary part for $\eps > 0$:
\beq\label{eq:ImPart1}
-\eps  ( 1 +  \Psi_j^T [ {(\widetilde{H}_L^N(j) -E )^2 + \eps^2} ]^{-1} \Psi_j ) < 0 .
\eeq
Consequently, we may apply Lemma \ref{lemma:ResolvInt1} to the right side of \eqref{eq:schur3} to obtain the integral representation:
\beq\label{eq:IntegralRep1}
\left[  v_{jj} - z - \Psi_j^T \widetilde{R}_{N,L} (z)  \Psi_j \right]^{-1 }
= \int_0^\infty  e^{-i \lambda \left(v_{jj} -z - \Psi_j^T \widetilde{R}_{N,L,j} (z) \Psi_j \right) }\,d\lambda.
\eeq
Substituting this representation into \eqref{eq:differenceSchur1} and taking the expectation, we find the following representation of  \eqref{eq:differenceSchur1}:
\beq\label{eq:differenceSchur2}
\E \left\{ \int_0^\infty  e^{-i \lambda (v_{jj}-E-i\eps) }
\left(e^{i \lambda \left(\Psi_j^T  \widetilde{R}_{N,L,j} (z)  \Psi_j\right) } - e^{i \lambda \left(\Psi_j^T  \widetilde{R}_{M,L,j} (z)
 \Psi_j\right) }\right) \right\} \,d\lambda.
\eeq
The second factor containing the difference of the two exponentials of the matrix elements of the resolvents  is independent of $v_{jj}$. Since the probability measure is a product measure,  we use Fubini's Theorem to bring the expectation with respect to $v_{jj}$ inside the $\lambda$ integral, to get
\beq\label{eq:differenceSchur3}
 \E_{v_{jj}^\perp} \left\{  \int_0^\infty  ~C(\lambda)
e^{-i \lambda (E - i \eps) }\left(e^{i  \lambda \left(\Psi_j^T  \widetilde{R}_{N,L,j} (z) \Psi_j\right)}
  - e^{ i   \lambda \left(\Psi_j^T   \widetilde{R}_{M,L,j} (z)  \Psi_j\right) }\right) ~d\lambda \right\}, 
\eeq
where 
\beq\label{eq:characteristic1}
C(\lambda) = \E_{v_{jj}} \left\{ e^{- i  \lambda v_{jj} }\right\}
=\int_\R ~ e^{-i \lambda  v_{jj} } \rho(v_{jj}) \,dv_{jj} = \sqrt{2 \pi} ~\widehat{\rho}(\lambda),
\eeq
is the characteristic function of $\rho$. 
 Under Assumption 1, the Fourier
transform $\widehat{\rho}$  has the bound
\beq\label{eq:DecayCharac1}
 | \widehat{\rho}(\lambda) | \leq \frac{C}{1+|\lambda|^2} ,
\eeq
for some finite constant $C > 0$.

\noindent
3. Due to the decay of the characteristic function of $\rho$ in \eqref{eq:DecayCharac1}, the expectation of the difference of the matrix elements in \eqref{eq:differenceSchur1} may be bounded as
\bea\label{eq:differenceSchur4}
\lefteqn{ | \E \{ \left\langle e_j, {R}_{N,L}(z)  e_j \right\rangle - \left\langle e_j,  R_{M,L}(z)  e_j \right\rangle \} | } \nonumber \\
 & \leq & \E_{v_{jj}^\perp} \left\{  \int_0^\infty \frac{C}{1+|\lambda|^2}
\left|\left(e^{i  \lambda \left(\Psi_j^T \ \widetilde{R}_{N,L,j} (z) \Psi_j\right) } - e^{i \lambda \left(\Psi_j^T  \widetilde{R}_{M,L,j} (z) 
 \Psi_j\right) }\right)\ \right|\, d\lambda \right\}. 
\eea
Now applying Lemma \ref{lemma:duhamel2} to the integrand in \eqref{eq:differenceSchur4}, we have an upper bound of
\beq\label{Eq:differenceSchur6}
 2^{1-s} C \left( \int_0^\infty \frac{\lambda^s}{1+|\lambda|^2} ~d \lambda \right)
   \E_{v_{jj}^\perp} \left\{  \left|\Psi_j^T   \widetilde{R}_{N,L} (z)  \Psi_j
- \Psi_j^T  \widetilde{R}_{M,L} (z) \Psi_j \right|^s \right\}.
\eeq
Since the $\lambda$-integral is convergent and the expectation is independent of $v_{jj}$, we obtain the desired bound
\beq\label{eq:differenceSchur7}
 C_{s, \rho} \E_{} \left\{ \left|\Psi_j^T  \widetilde{R}_{N,L} (z) \Psi_j
- \Psi_j^T  \widetilde{R}_{M,L} (z)  \Psi_j \right|^s \right\} ,
\eeq
for  a finite constant $C_s > 0$ depending on $s$ and the constant in \eqref{eq:DecayCharac1}.
\end{proof}

\begin{remark}
Theorem \ref{thm:fractionalMoment1} is a simplification of \cite[Theorem 2.2]{dkm} developed for random Schr\"odinger operators with single-site probability measures having compact support.
The theorem of \cite{dkm} applies to fixed bandwidth random band matrices for which the probability density $\rho$ has compact support.
 In this setting, the theorem can be stated as follows:

\begin{theorem} \cite[Theorem 2.2]{dkm} \label{thm:Thm4dkm}
Let $H_L^N$ be a symmetric random band matrix with a fixed bandwidth, and the common probability density compactly supported and smooth with support in $(0,R)$ for some $R$.  Then for any $z\in\C^+$,
\[ \left| \E \left \{ \left\langle e_0, \left( (H_L^N-z)^{-1} - (H_L^M-z)^{-1}\right) e_0 \right\rangle \right\} \right|
\]
\[\leq C \E_{\omega_0^\perp}\left\{ \int d\omega_0\, \chi_{\left[-1-\frac{5R}{2}, - \frac{R}{2}\right]}\left(\omega_0\right)\left|  \left\langle e_0, \left( (H_L^N-z)^{-1} - (H_L^M-z)^{-1}\right) e_0 \right\rangle \right|^s \right\}.
\]
\end{theorem}

\noindent
One of the main advances of Theorem \ref{thm:fractionalMoment1} is the removal of the compact support condition on the  probability measure.

\end{remark}

\section{A localization bound for fixed-width random band matrices}\label{sec:localBound1}

In this section, we prove that the expectation in \eqref{eq:differenceSchur7} decays exponentially with respect to $M$ (for $N \geq M)$ provided the matrix elements of the resolvents of  reduced random matrices $\widetilde{H}_L^N(j)$ and 
$\widetilde{H}_L^M(j)$ satisfy an exponential fractional moment bound. These bounds on the resolvents of the reduced random matrices are proved in Proposition \ref{prop:fracMomentreduced1} provided the original random matrices $H_L^N$ and $H_L^M$ satisfy a fractional moment bound. For fixed-width random band matrices, these fractional monent bounds follow from  the extension of Schenker's result,  Theorem \ref{thm:fractMoment1},  presented in Theorem \ref{thm:locAllE1}.  


\subsection{A localization bound for the difference of fixed-width RBM}\label{subsec:fmb1}

In this section, we prove an exponentially-decaying upper bound on the right side of \eqref{eq:fracMoment2}, assuming a localization bound on reduced RBM proved in section \eqref{subsec:reducedRM1}.
As in section \ref{sec:fracMoment1}, we write  $\widetilde{R}_{N,L,j} (z) = (\widetilde{H}_L^N(j) - z)^{-1}$, where $\widetilde{H}_L^N(j)$ is obtained from $H_N^L$ by setting the $j^{\rm th}$-row and $j^{\rm th}$-column equal to zero. 

\begin{theorem}\label{thm:fractionalMoment2}
Suppose $N  \geq  M \gg L \geq 0$.  As in the proof of Theorem \ref{thm:fractionalMoment1}, we let $P = |{e_j}\rangle \langle e_j|$ and $Q=1-P$ acting on $\C^{2N+1}$ with the understanding that $\C^{2M+1}$ is embedded in $\C^{2N+1}$ as described in section \ref{sec:fracMoment1}.
By the same embedding, the vector $\Psi_j= Q H_L^M Pe_j$ can be identified with $Q H_L^N Pe_j$ by the same embedding. We also write
$\hn (j)= Q H_L^N(j) Q$ and $\hm(j) = Q H_L^M(j) Q$ as in Theorem \ref{thm:fractionalMoment1}.
Finally, suppose $0 < s<1/9$.  Then, for any $\epsilon > 0$, there exist finite constants $C_{L,s, \rho} > 0$ and $\gamma_{L,s}>0$, depending on $s$, the band width $L$, and the single-site density $\rho$, such that for $| j| \leq M - L$,
\beq\label{eq:locFMM1}
\E \left\{ \left|\left\langle \Psi_j,  \widetilde{R}_{N,L,j} (z) \Psi_j\right\rangle
- \left\langle \Psi_j,  \widetilde{R}_{M,L,j} (z) \Psi_j \right\rangle \right|^s \right\}  \leq  C_{L,s, \rho}   e^{-\gamma_{L,s, \rho} (M- |j|)}.
\eeq
\end{theorem}





\begin{proof}
Applying the second resolvent identity to the left side of \eqref{eq:locFMM1}, we must estimate
\beq\label{eq:resolvDiff1}
\E \left\{  \left| \left\langle \Psi_j,  \widetilde{R}_{N,L,j} (z) 
 \left(\hn(j) - \hm(j) \right)
 \widetilde{R}_{M,L,j} (z)  \Psi_j\right\rangle\right|^s \right\}.
\eeq
We recall that $[\hm(j)]_{km}=[\hn(j)]_{km}=v_{km}$, for $-M \leq k, m \leq M$, $|k-m|\leq L$, so that 
$[ \hn(j) - \hm(j) ]_{km} = 0$, for all $-M \leq k,m \leq M$.
All matrix elements $[\hm(j)-\hn(j)]_{km}$ are zero\emph{ except} for a subset of indices $(k,m)$ in either of the following two cap regions: For indices near $(-N,-N)$, the cap region $\mathcal{C}^-_{N,M}$  is described by
\bea\label{eq:upperCap1}
\mathcal{C}^-_{N,M}  & := & \{ (k, m) ~|~  -N \leq k \leq -M - 1 ~~ {\rm and} ~~ k \leq m \leq k + L \}  \nonumber \\
       && \cup ~ \{ (k,  m ) ~|~  m \leq k \leq m + L ~~ {\rm and} ~~ -N  \leq m \leq -M - 1 \} ,
\eea
and, for near indices $(N,N)$, the cap region $\mathcal{C}^+_{N,M}$ is
\bea\label{eq:lowerCap1}
\mathcal{C}^+_{N,M}  & := & \{ (k,m) ~|~   M+1 \leq k \leq N ~~{\rm and} ~~ k-L \leq m \leq k   \} \nonumber \\
 & &  \cup
 \{ (k,m) ~|~ m - L \leq k \leq  m ~~{\rm and} ~~ M+1 \leq m \leq  N \} .
\eea
We also recall that the vector $\Psi_j$ is defined by
\beq\label{eq:psi1}
 \Psi_j = \sum_{\substack{i=j- L,\dots, j +L \\ i \neq j}}
[H_L^N ]_{ij}  e_i.
\eeq
Let $\mathcal{S}_j := \{ (u,v,k,\ell) \}$ denote the set of indices $\{(u,v,k,\ell)\}$ satisfying the conditions:
\begin{itemize}
\item $(k,\ell) \in \mathcal{C}^-_{N,M} \cup \mathcal{C}^+_{N,M}$
\item $-L+j \leq u, v \leq L+j, ~~~u,v \neq j$;
\end{itemize}
Then, substituting this into \eqref{eq:resolvDiff1}, we obtain
\bea\label{eq:resolvDiff2}
\lefteqn{ \left\langle \Psi_j, \widetilde{R}_{N,L,j}(z)  \left(\hm(j) - \hn(j) \right)
  \widetilde{R}_{M,L,j}(z)  \Psi_j\right\rangle }  \nonumber \\
 &  = &
\sum_{(u,v,k,\ell) \in \mathcal{S}_j} v_{uj}v_{jv}v_{k\ell}
  \left\langle e_v,  \widetilde{R}_{N,L,j} (z) e_k \right\rangle \left\langle e_\ell,
 \widetilde{R}_{M,L,j}(z) e_u\right\rangle.
\eea
In order to take the $s^{\rm th}$-power of \eqref{eq:resolvDiff2}, we use the the basic inequality 
\beq\label{eq:BasicIneq1}
\left|\sum_{i=1}^N a_i \right|^s \leq \sum_{i=1}^N |a_i|^s, 
\eeq
for $0 < s < 1$.
Then, the expectation of the $s^{\rm th}$-power of the left side of \eqref{eq:resolvDiff2} satisfies the inequality
\bea\label{eq:resolvDiff3}
\lefteqn{ \E \left\{  \left| \left\langle \Psi_j, \widetilde{R}_{N,L,j}(z)  \left(\hm(j) - \hn(j) \right)
  \widetilde{R}_{M,L,j}(z)  \Psi_j\right\rangle \right|^s \right\} }  \nonumber \\
 & \leq &  \sum_{(u,v,k,\ell) \in \mathcal{S}_j} \E  \left\{ \left| v_{uj}v_{jv}v_{k\ell} \left\langle e_v, \widetilde{R}_{N,L,j}(z) e_k\right\rangle
\left\langle e_\ell, \widetilde{R}_{M,L,j} (z) e_u\right\rangle\right|^s \right\}.
\eea
%
Using a generalized Holder inequality, we bound each term in the sum by
\beq\label{eq:resolvDiff4}
 \left(\E \{ |v_{uj}v_{jv}v_{k\ell}|^{3s} \} \right)^{1/3}
\E \left\{  \left|\left\langle e_v, \widetilde{R}_{N,L,j} (z)  e_k\right\rangle\right|^{3s}\right\}^{1/3}
 \E \left\{ \left|\left\langle e_\ell, \widetilde{R}_{M,L,j}(z) e_u \right\rangle \right|^{3s}\right\}^{1/3}.
\eeq
The first term is bounded by a constant under the conditions on the moments of the random variables in Assumption 1. Exponential bounds on the second and third expectations involving the resolvents of $\widetilde{H}_L^N(j)$ and 
$\widetilde{H}_L^M(j)$ are given in Proposition \ref{prop:fracMomentreduced1}, proved section \ref{subsec:reducedRM1}.
In Proposition \ref{prop:fracMomentreduced1}, the parameter $s$ must be resrticted to $0 < s < \frac{1}{3}$, see \eqref{eq:diffExp42}. Consequently, expression \eqref{eq:resolvDiff4} limits $s$ to the range $0 < s < \frac{1}{9}$.
Although $| \mathcal{S}_j| = \mathcal{O}(\max ((N-M)L, L^2))$, we note that only $\mathcal{O}(L^2)$ terms contribute to the sum on the right in \eqref{eq:resolvDiff3} because, in addition to the constraint of being in $\mathcal{S}_j$, the index 
$\ell$ must satisfy $||\ell| - M| \leq L$,  and the index $k$ must satisfy $| k - \ell| \leq L$. Consequently, we must have $||k| - M | \leq 2 L$. Hence, there are only $\mathcal{O}(L^2)$ nonzero terms in the sum on the right in \eqref{eq:resolvDiff3}.    As a result, we obtain the bound  
\beq\label{eq:ModResolv1}
\E \left\{ \left|\left\langle \Psi_j,  \widetilde{R}_{N,L,j} (z) \Psi_j\right\rangle
- \left\langle \Psi_j,  \widetilde{R}_{M,L,j} (z)   \Psi_j \right\rangle\right|^s \right\} 
\leq  C_{L,s, \rho} e^{-\gamma_{L,s} \left(M- |j|  \right)}  ,
\eeq 
proving \eqref{eq:locFMM1}.
\end{proof}




\subsection{A localization bound for reduced RBM}\label{subsec:reducedRM1}

We prove that the resolvents of the reduced random band matrices $\widetilde{H}_L^N(j)$ have fractional moments that decay exponentially if the original random band matrices exhibit the same behavior. This will follow from Theorem \ref{thm:locAllE1} which is valid for fixed bandwidth random matrices. Recall that $\widetilde{H}_L^N (j)$ is not an  random band matrix itself 
since the elements of the $j^{\rm th}$-row and $j^{\rm th}$-column are equal to zero.

\begin{prop}\label{prop:fracMomentreduced1}
For $j \in \{ -N, \ldots, N \}$, let $\widetilde{H}_L^N(j)$ be the matrix obtained from $H_L^N$  by setting the $j^{th}$-row and the $j^{\rm th}$-column equal to $0$.  Suppose the indices $(i,k)$ satisfy $|i-j|\leq L$ and $||k| - N|\leq L$ and that $i \neq j$.   Then, there exist finite, positive, constants $C_{L,s, \rho}$ and $\gamma_{L,s}$ such that for $0 < s < \frac{1}{3}$,
\begin{equation}
\E \left\{  \left|\left\langle e_i, \widetilde{R}_{N,L,j}(E + i\eps) e_k \right\rangle\right|^{s} \right\} \leq 
C_{L,s, \rho}e^{-\gamma_{L,s}(N-|j| )}
\end{equation}
\end{prop}

\begin{proof}
1.  We compare the resolvent $\widetilde{R}_{N,L,j}(z)$ for $\widetilde{H}_L^N (j)$ with the resolvent $R_{N,L}(z)$ for the original band matrix $H_L^N$, with $z=E + i \epsilon$. The second resolvent formula yields 
\begin{eqnarray}\label{eq:diffExp1}
\E \left\{  \left|\left\langle e_i,  \widetilde{R}_{N,L.j}(z) e_k\right\rangle\right|^{s} \right\}
&\leq &  \E \left\{ \left|\left\langle e_i,  {R}_{N,L}(z) e_k\right\rangle\right|^{s}  \right\} \nonumber \\ 
&  & +
 \E \left\{ \left|\left\langle e_i,   \left[ \widetilde{R}_{N,L,j}(z) -   {R}_{N,L}(z) \right]
e_k\right\rangle\right|^{s} \right\}. 
\end{eqnarray}
The first term on the right in \eqref{eq:diffExp1} decays exponentially like $|i-k|$ by Theorem \ref{thm:locAllE1}.

\noindent
2. The second term on the right in \eqref{eq:diffExp1} may be written as
\beq\label{eq:diffExp3}
 \E \left\{  \left|\left\langle e_i,   \left[ \widetilde{R}_{N,L,j}(z)( H_L^N - \widetilde{H}_L^N(j))  {R}_{N,L}(z) \right] e_k\right\rangle\right|^{s} \right\} .
\eeq
The difference $(H_L^N- \hn (j) )$ is the matrix consisting only of the $j^{\rm th}$-row and $j^{\rm th}$-column of $H_L^N$, for $j  \in \{ -N,  \ldots, N \}$. 
We define an  index set $\mathcal{T}_j$ by
\beq\label{eq:tauIndex1}
\mathcal{T}_j :=\left\{(m,n)\,:\, \begin{array}{cl} & m=j\ and \  \max ( -N, j-L) \leq n \leq  \min ( N, j+L) \\ or &
n=j\ and \   \max ( -N, j-L)  \leq m \leq  \min (N, j+L) \end{array} \right\}.
\eeq
 Then, the matrix elements of  difference $(H_L^N- \hn (j) )$ are
\[ [ H_L^N-\hn (j) ]_{mn}=
\left\{\begin{array}{ccc}
v_{mn} & if & (m,n) \in \mathcal{T}_j \\
0 & otherwise \end{array}\right.\]
Using this and the basic inequality \eqref{eq:BasicIneq1}, we obtain an upper bound for \eqref{eq:diffExp3}:
\beq\label{eq:diffExp41}
\sum_{(m,n)\in \mathcal{T}_j} \E \{
 | \langle e_i,   \widetilde{R}_{N,L,j}(z) e_m \rangle |^{s}  |v_{mn} |^{s} 
\left|\left\langle e_n,  {R}_{N,L}(z) e_k\right\rangle\right|^{s} \}.
\eeq
With another use of the generalized Holder inequality, we have the bound
\beq\label{eq:diffExp42}
\sum_{(m,n)\in \mathcal{T}_j} ( \E \{ | \langle e_i, \widetilde{R}_{N,L,j}(z) e_m \rangle |^{3s} \} )^{1/3} 
 (\E  |v_{mn} |^{3s} )^{1/3}
 ( \E \{ | \langle e_n, R_{N,L}(z) e_k \rangle |^{3s} \} )^{1/3} .
\eeq
We restrict $s$ so that $0 < s < \frac{1}{3}$. 
We bound the first term in the product by a constant using the {\it a priori} bound from Proposition \ref{proposition:SpectralAve1} when $m$ is not equal to $j$. For these cases, we note that the proof of Proposition \ref{proposition:SpectralAve1} requires averaging with respect to the diagonal elements only. When $m=j$, we note that the vector $e_j$ is in the kernel of $\widetilde{H}_L^N(j)$ so that $\widetilde{R}_{N,L,j}(z) e_j = - z^{-1} e_j$, valid as $\Im z \neq 0$. Substituting this equation into the first matrix element of \eqref{eq:diffExp42}, and using the fact that $i \neq j$, we find that the term is zero due to orthogonality. 
 The second term in the product is bounded by the first moment of the random variables and hence is bounded by Assumption 1.  The third term decays exponentially by Theorem \ref{thm:locAllE1}. 

\noindent
3. As a consequence, we obtain the bound
\bea\label{eq:diffExp6}
\E \{  | \langle e_i,  \widetilde{R}_{N,L,j}(z) e_k \rangle|^{s} \} & \leq &  \E \{  | \langle e_i,  {R}_{N,L}(z) e_k \rangle|^{s} \}  \nonumber \\
 & & +  \sum_{(m,n)\in\mathcal{T}_j} C_{L,s, \rho}  \E \{ | \langle e_n, R_{N,L}(z) e_k \rangle |^{3s} \}^{1/3} \nonumber \\
  & \leq & \widetilde{C}_{L,s,\rho} e^{- \alpha_{L,s} |i - k|} +    \sum_{(m,n)\in\mathcal{T}_j} C_{L,s, \rho} 
       e^{- \alpha_{L,s} |n-k|}  \nonumber \\
 & \leq &  C_{L,s,\rho}  \left(   \sum_{(m,n)\in\mathcal{T}_j} e^{-\alpha_{L,s} \min( |n-k|, |i-k| ) } \right) . 
\eea
Recalling the definition of $\mathcal{T}_j$ in \eqref{eq:tauIndex1}, we have,
\begin{align}
    &| n - j | \leq L \nonumber\\
    &||k| - N| \leq L\nonumber\\
\end{align}
    and so, as $|j| \leq N$, we obtain
\bea\label{eq:Indices1}
0 < N - |j| & \leq & | N - |k|| + | |k| -  n| + |n-|j|| \nonumber \\
          &      \leq & 3L + | n - k| .
\eea
A similar bound holds for $|i-k|$ as $|i-j| \leq L$. 
Since $| \mathcal{T}_j| = \mathcal{O}(L)$, there exist finite constants $\gamma_{L,s} > 0$ and $C_{L,s, \rho} > 0$, so that 
\beq\label{eq:diffExp7}
\E \{  | \langle e_i,  \widetilde{R}_{N,L.j}(z) e_k \rangle |^{s} \} \leq C_{L,s, \rho} e^{- \gamma_{L,s} ( N - |j| ) },
\eeq
completing the proof.
\end{proof}

%
%


\section{Convergence and smoothness of density of states function}\label{sec:ptwConvDOS1}
\setcounter{equation}{0}


In this section, we  prove the uniform convergence of the local density of states functions $n_L^N(E)$ to $n_L^\infty (E)$. 
We also prove that the integrated density of states $N_L^\infty(E)$ is differentiable as many times as the probability density $\rho$. Throughout this section, we use the representations of the local density of states functions $n_L^N$ given in \eqref{eq:dosDefn1} of Proposition \ref{prop:dosDefn1}: 
\begin{equation}
    n_L^N(E)= \lim_{\eps \rightarrow 0^+} \frac{1}{2N+1} \frac{1}{\pi} \sum_{j=-N}^N \E   \{ \Im
    \left\langle e_j, (H_L^N - E - i \eps)^{-1}
    e_j\right\rangle \} ,
\end{equation}
We also recall that the infinite-volume DOSf, $n_L^\infty(E)$, has the following representation following from the Birkhoff Ergodic Theorem:
\bea\label{eq:repDosInf1}
    n_L^\infty(E) & =  &  \lim_{\eps \rightarrow 0^+} \frac{1}{\pi} \E   \{ \Im
    \left\langle e_0, (H_L^\infty - E - i \eps)^{-1}    e_0\right\rangle   \} \nonumber \\
 & = &     \lim_{\eps \rightarrow 0^+} \frac{1}{2N+1} \frac{1}{\pi} \sum_{j=-N}^N \E  \{ \Im   \left\langle e_j, (H_L^\infty - E - i \eps)^{-1}    e_j \right\rangle \} .
\eea

\subsection{Uniform convergence of the local density of states function}\label{subsec:unifConvDOS1}

The main result of this section is the following theorem that is essential for the identification of the intensity measure of the Poisson point process in section \ref{sec:poisson1}.

\begin{theorem}\label{thm:dosLimit1}
Let $H_L^N$ be a real, symmetric random band matrix with fixed bandwidth $2L+1$ with entries satisfying Assumption 1. 
Then the local density of states function $n_L^N$ converges uniformly on $\R$ to the density of states function $n_L^\infty$. 
\end{theorem}

\begin{proof}
1. For $z = E + i \epsilon \in \C^+$, we estimate the difference
\beq\label{eq:dosLimit2}
n_L^N (E) - n_L^\infty (E)  =   
 \frac{1}{2N+1} \frac{1}{\pi}  \sum_{j=-N}^N  \lim_{\eps \rightarrow 0^+} \E \{ \Im [ \langle e_j,
 R_{N,L}(z)  e_j  \rangle - \langle e_j,   R_{\infty,L} (z) e_j \rangle ]   \}.
\eeq
It follows from Theorems \ref{thm:fractionalMoment1} and \ref{thm:fractionalMoment2}, that
\beq\label{eq:dosLimit3}
\E  \left\{ \left|  \left\langle e_j,
\left(H_L^N - z\right)^{-1} e_j\right\rangle -
\left\langle e_j, \left(H_L^\infty - z\right)^{-1} e_j\right\rangle \right| \right\}
\leq C_{L,s,\rho} e^{-\alpha_{L,s} (N- |j|) }.
\eeq
In addition, by part 2 of Proposition \ref{proposition:SpectralAve1}, we have the uniform bound
\beq\label{eq:dosUnifBd1}
| \E \{ \Im \langle e_j, (H_L^N - z )^{-1} e_j  \rangle
- \Im \langle e_j,  (H_L^\infty - z )^{-1} e_j  \rangle \} | \leq C.
\eeq
The constants in both \eqref{eq:dosLimit3} and \eqref{eq:dosUnifBd1} are  independent of $N$ and $z$.

\noindent
2. We take $0<\alpha<1$ and divide the $j$ sum in \eqref{eq:dosLimit2} into two terms: $|j| \leq N - N^\alpha$, and $N- N^\alpha < |j| \leq N$. For the first regime, it follows from \eqref{eq:dosLimit3} that
\bea\label{eq:dosLimit4}
 \sum_{ |j| \leq N - N^\alpha} \E \left\{ \left| \left\langle e_j, R_{N,L}(z)  e_j\right\rangle - \left\langle e_j,
 R_{\infty,L} (z)  e_j\right\rangle \right| \right\}
& \leq &  {C_1 (N-N^\alpha)} e^{-\alpha_{L,s}  N^\alpha} \nonumber \\
 & &
\eea
To control the second regime, we use the uniform bounds \eqref{eq:dosUnifBd1} and obtain
\beq\label{eq:dosLimit5}
\sum_{ N - N^\alpha < |j| \leq N } \E \left\{ \left| \left\langle e_j, R_{N,L}(z)  e_j\right\rangle - \left\langle e_j,
 R_{\infty,L} (z)  e_j\right\rangle \right| \right\}
\leq C_2  N^\alpha.
\eeq
Both constants $C_1$ and $C_2$ are uniform in $z$ and $N$. 
Returning to \eqref{eq:dosLimit2}, we have the upper bound
\bea\label{eq:dosLimit6}
\lefteqn{ \frac{1}{2N+1} \frac{1}{\pi} \sum_{ j= -N}^N \E \left\{ \left| \left\langle e_j, R_{N,L}(z)  e_j\right\rangle - \left\langle e_j,
 R_{\infty,L} (z)  e_j\right\rangle \right| \right\}  } \nonumber \\
 & \leq & 
  C_1  \frac{(N-N^\alpha)}{2N + 1} e^{-\alpha_{L,s} N^\alpha}
   + C_2 \frac{ N^\alpha}{2N + 1} ,
\eea
that vanishes as $N \rightarrow \infty$. 
Because of the uniformity of the convergence in \eqref{eq:dosLimit6}, we can interchange the $\eps$ and $N$ limits in the following and apply the uniform bound  \eqref{eq:dosLimit5} to obtain:
\bea\label{eq:dosLimit7} 
 \lefteqn{   \lim_{N\to\infty} | n_L^N(E) -  n_L^\infty (E) | } \nonumber \\ 
    & = & \lim_{\eps\rightarrow 0^+} \lim_{N\to\infty} \frac{1}{2N+1} \frac{1}{\pi}  \left| \sum_{j=-N}^N  \E\left\{
     \Im \left\langle e_j, R_{N,L}(z)    e_j\right\rangle - \Im \left\langle e_j, R_{\infty,L}(z) e_j\right\rangle
    \right\} \right| \nonumber \\
 & = & 0,  
\eea
proving the result. 
\end{proof}

\subsection{Smoothness of the DOSf}\label{subsec:smoothDosf1}

The estimates of section \ref{sec:fundamentals1} and \ref{sec:localBound1} allow us to prove that the integrated density of states $N_L^\infty (E)$ is smooth in the case of Gaussian random variables for which the probability density $\rho$ is given by  \eqref{eq:gauss1}. 
In general, for a probability density $\rho$ satisfying Assumption 1 with $\rho \in C^k(\R)$, we will prove that the IDS $N_L^\infty \in C^k(\R)$.  Inspiration and some of the techniques for the following proofs come from \cite{dkm}. As a first step, we state the following lemma which appears in \cite[Lemma A.1]{dkm}.

\begin{lemma}\label{lemma:smoothCondition1}
Consider a positive function $f \in L^1(\R,dx)$ and an interval $J\subset \R$.  For $z \in \C^+$, the Borel transform of $f$
is defined by
$$
F(z) :=\int \frac{1}{x-z} f(x) ~dx  , ~~~ {\rm for} ~~ z \in \C^+ .
$$
 Then, if for some $m\in\N$,
\beq\label{eq:smoothCondition2}
    \sup_{z\in \C^+,\ \Re(z)\in J} \left|\frac{d^m}{dz^m}\Im(F(z))\right|< \infty ,
\end{equation}
then, we have
\begin{equation}\label{eq:smoothCondition3}
    \mbox{\rm ess sup}_{x\in J} \left|\frac{d^m}{dx^m}f(x)\right| < \infty .
\end{equation}
\end{lemma}

We will apply this lemma with $f=n_L^\infty$, the infinite-volume density of states function. This smoothness will follow from 
bounds on the derivatives of Borel transform of the density of states measure that will be uniform in $\C$.
We recall that from \eqref{eq:repDosInf1} and the ergodicity of the infinite-volume operator $H_L^\infty$, we have the representation:
\begin{equation}
    n_L^\infty(E)=\lim_{\eps\rightarrow 0^+} \frac{1}{\pi} \E \{ \Im\langle e_0, (H_L^\infty -E-i\eps)^{-1} e_0 \rangle \}.
\end{equation}
To simplify the presentation, we prove $n_L^\infty \in \mathcal{C}^1(\R)$ provided $\rho$ satisfies Assumption 1.
 Higher-order derivatives are treated as in \cite{brodie1} and \cite{dkm}, and we summarize this in Corollary \ref{cor:anyDerivativeIDS1}.

\begin{theorem}\label{thm:smoothDos1}
 Let $H_L^\infty$ be a fixed-width random band matrix on $\ell^2(\Z)$ with entries satisfying Assumption 1. Then the corresponding density of states function $n_L^\infty\in \mathcal{C}^1(\R)$.
\end{theorem}

\begin{proof} 
\noindent
1. By Lemma \ref{lemma:smoothCondition1}, we must prove that    
\beq
    \sup_{z \in \C^+} \left|\frac{d}{dz}
    \E  \{ \Im \langle e_0, (H_L^\infty-z)^{-1} e_0 \rangle  \}
    \right|
    < \infty.
\eeq
From Theorem \ref{thm:dosLimit1}, 
we have that 
    \begin{equation} 
    \E \{ \langle e_0, (H_L^N-z)^{-1} e_0\rangle \}
    \to \E  \{ \langle e_0, (H_L^\infty-z)^{-1}
    e_0\rangle \}
    \end{equation}
    as $N\to\infty$ uniformly for $z \in \C^+$.  Furthermore, since the Green's functions are analytic for $z \in \C^+$, this implies that the derivatives also converge uniformly on compact subsets of the upper half-plane.
    We can therefore write the infinite-volume Green's function at fixed $z\in \C^+$ as the following telescoping series:
    \bea\label{eq:telescoping}
    \E \{  \Im  \langle e_0, R_{\infty,L}(z) e_0 \rangle \}     &= & \sum_{M=N}^\infty
    \left[ \E \{ \Im \langle e_0, R_{M+1,L}(z) e_0 \rangle -  \E  \{ \Im \langle e_0, R_{M,L}(z) e_0 \rangle \} \right] \nonumber \\
    & & +\E \{ \Im \langle e_0, R_{N,L}(z) e_0 \rangle \}.
    \eea
    From Proposition \ref{prop:smoothLIDS1},
    \begin{equation}
\mbox{ess sup}_{z \in \C^+} \frac{d}{dz}\E \{ \Im \langle e_0,  R_{N,L}(z) e_0 \rangle \}
\end{equation}
    is finite.  Thus it remains to bound 
    \begin{equation}
        \frac{d}{dz}\E \{ \Im \left[ \langle e_0,  R_{M+1,L}(z) e_0 \rangle -  \langle e_0,  R_{M,L}(z)  e_0 \rangle \right] \} ,
    \end{equation}
    and prove that it is summable.


\noindent
2. For $z=E+i\eps \in \C^+$, the Cauchy-Riemann equations show that  $ \frac{d}{dz}\langle e_0,  R_{M,L}(z) e_0 \rangle$ 
    can be written in terms of derivative with respect to $E$ of the real and imaginary parts of the Green's function.  Thus it suffices to obtain estimates for
    \begin{equation}\label{eq:smooth1}
    \frac{d}{dE} 
    \E \left\{ \langle e_0,  R_{M+1,L}(E + i\eps) e_0 \rangle  - \langle e_0,  R_{M,L}(E + i\eps) e_0 \rangle \right\} .
     \end{equation}
 We denote by $\E_{\rm off}$ the expectation with respect to the off-diagonal elements of $H_L^M$ and factor out  the expectation with respect to the diagonal terms  in \eqref{eq:smooth1}. We write the expectation in \eqref{eq:smooth1} as
    \begin{align}
    & \E \left\{ \langle e_0,  R_{M+1,L}(z) e_0 \rangle - \langle e_0,  R_{M,L}(z) e_0 \rangle \right\} \nonumber\\
    &=\E_{\rm off} \left\{ \int_{I_\rho} ~ \prod_{i=-M-1}^{M+1} dv_{ii} \rho(v_{ii}) 
    \left[ \langle e_0,  R_{M+1,L}(z) e_0 \rangle - \langle e_0,  R_{M,L}(z) e_0 \rangle \right] \right\}, 
    \end{align}
where $\rho$ is the probability density function with support $I_\rho \subset \R$ satisfying Assumption1. 
Proceeding as in the proof of Proposition \ref{prop:smoothLIDS1}, we make a change of the diagonal variables by $v_{ii} \rightarrow v_{ii} - E$
and obtain:
    \bea\label{eq:derivSmooth1}
   \lefteqn{ \frac{d}{dE}   \E \left\{ \langle e_0,  R_{M+1,L}(z) e_0 \rangle - \langle e_0,  R_{M,L}(z) e_0 \rangle \right\}  } \nonumber \\
 & = &    \sum_{i=-M-1}^{M+1} \int_{I _\rho - E} \rho'(v_{ii} + E) 
    \prod_{j\neq i}\rho(v_{jj} + E)      
  \E_{\rm off}  \left\{ \langle e_0,   R_{M+1,L}(i \eps) e_0 \rangle - \langle e_0,  R_{M+1,L}(i \eps)  e_0 \rangle\right\}  \nonumber \\
  & &
\eea
Undoing this change of variables, we arrive at
\bea\label{eq:Smooth1v2}
  \lefteqn{ \frac{d}{dE}   \E \left\{ \langle e_0,  R_{M+1,L}(z) e_0 \rangle - \langle e_0,  R_{M,L}(z) e_0 \rangle \right\}  } \nonumber \\
 & = &  \sum_{i=-M-1}^{M+1} \E_{{v_{ii}}^\perp} \left\{  \int_{I_\rho}  \rho^\prime (v_{ii}) \left[\langle e_0,  R_{M+1,L}(z) e_0 \rangle - \langle e_0,  R_{M,L}(z) e_0 \rangle\right] \right\}  ,
\eea
where $\displaystyle \E_{{v_{ii}}^\perp}$ is the expectation of all random variables except for $v_{ii}$.
We now proceed as in the proof of Theorem \ref{thm:fractionalMoment1}.
In particular, in analogy to \eqref{eq:differenceSchur1}-\eqref{eq:differenceSchur2}, we have 
\begin{align}\label{eq:smooth4}
   &     \int_{I} \rho^\prime (v_{ii}) \left[ \langle e_0,  R_{M+1,L,j}(z) e_0 \rangle - \langle e_0,  R_{M,L,j}(z) e_0 \rangle \right]   \nonumber \\
   &= \int_{I} dv_{ii} \rho^\prime (v_{ii}) \int_0^\infty  e^{-i  \lambda (v_{00}- z) } 
  \left( e^{i \lambda \left( \Psi_0^T \widetilde{R}_{M+1,L,0}( z) \Psi_0\right)} - e^{i \lambda \left(\Psi_0^T \widetilde{R}_{M,L,0}(z)  \Psi_0\right) }  \right) \,d\lambda.
\end{align}

\noindent
3. We treat the cases $i = 0$ and $i \neq 0$ separately. If $i=0$, we obtain the Fourier transform of $\rho^\prime$ from the $\lambda$-integral in \eqref{eq:smooth4}. Returning to the expectation, we obtain:
\begin{align}\label{eq:smooth5}
  \E_{{v_{00}}^\perp}  \left\{ \int_0^\infty \widehat{\rho'}(\lambda) e^{i  \lambda z } \left(e^{i \lambda \left(\Psi_0^T \widetilde{R}_{M+1,L,0}(z) \Psi_0\right)} - e^{i \lambda \left(\Psi_0^T \widetilde{R}_{M,L,0}(z) \Psi_0\right) }\right)\,d\lambda \right\} .
\end{align}
Applying Lemma \ref{lemma:duhamel2}, we obtain the upper bound
\beq\label{eq:smooth6}
 \E_{{v_{00}}^\perp} \left\{  \int_0^\infty |\widehat{\rho'}(\lambda)| 2^{1-s} \lambda^s   
     \left|\Psi_0^T   \widetilde{R}_{M+1,L,0}(z)   \Psi_0
- \Psi_0^T   \widetilde{R}_{M,L,0}(z) \Psi_0 \right|^s ~  d\lambda   \right\}.
\eeq
By Assumption 1, $ \langle \lambda \rangle^m \widehat{\rho}(\lambda) \in L^\infty (\R)$, for $0 \leq m \leq k+1$.  Thus, for some finite constant $C_\rho > 0$, uniform on $\C^+$, the expression in \eqref{eq:smooth5} above is bounded by
\begin{equation}\label{eq:smooth7}
    C_\rho ~ \E_{{v_{00}}^\perp} \left\{ \left|\Psi_0^T \widetilde{R}_{M+1,L,0}(z)  \Psi_0
- \Psi_0^T \widetilde{R}_{M,L,0}(z) \Psi_0 \right|^s \right\} .
\end{equation}
We can now directly apply Theorem \ref{thm:fractionalMoment2} (with $N = M+1$) and obtain the bound
\beq\label{eq:smooth71}
   \left|  \E_{{v_{00}}^\perp}  \int_{I_\rho}  dv_{00} ~\rho'(v_{00})  \left\{ \langle e_0, R_{M+1, L}(z) e_0 \rangle -  \langle e_0, R_{M,L}(z)  e_0 \rangle \right\} \right|     \leq   C_{L,s,\rho} e^{-\alpha_{L,s} M}.
\eeq

\noindent
4. If $i\neq 0$ in \eqref{eq:smooth4}, we first integrate \eqref{eq:smooth4} over $v_{00}$. 
Using Lemma \ref{lemma:duhamel2} and the decay of $\widehat{\rho}$ in Assumption 1, we bound \eqref{eq:smooth4} by 
\beq\label{eq:smooth8}
\E_{(v_{00},v_{ii})^\perp} \left\{ \int_{I_\rho}  dv_{ii} \left|\rho'(v_{ii})\right| \left|\Psi_0^T  \widetilde{R}_{M+1,L,0}(z)  \Psi_0 - \Psi_0^T  \widetilde{R}_{M,L,0}(z)  \Psi_0 \right|^s \right\} . 
\eeq
We now let
\begin{align} C_\rho^\prime = \int_\R |\rho'(x)| \, dx > 0,
\end{align}
which is finite by Assumption 1, and define a probability density 
\begin{equation}
    \widetilde{\rho}=\frac{1}{C_\rho'} | \rho^\prime |. 
\end{equation}
We let $\widetilde{\E}$ denote the expectation given by 
\beq\label{eq:expPrime1}
\widetilde{\E }\{ X \} :=  \E_{(v_{00},v_{ii})^\perp} \left\{ \int_{I_\rho}  dv_{ii} \widetilde{\rho}(v_{ii}) \{ X \} \right\} ,
\eeq
for a random variable $X$ depending on all $v_{ij}$.   
Consequently, we can bound the expectation in $\eqref{eq:smooth8}$ by 
\begin{equation}
C_\rho' ~ \widetilde{\E} \left\{  \left|\Psi_0^T \widetilde{R}_{M+1,L,0}(z) \Psi_0
- \Psi_0^T  \widetilde{R}_{M,L,0}(z)  \Psi_0 \right|^s \right\}  .
\end{equation}
%
%
The analysis in the proof of Theorem \ref{thm:fractionalMoment2} does not require identically distributed diagonal random variables. Consequently, we obtain
\bea \label{bound2}
{  \E_{(v_{00},v_{ii})^\perp}  ~ \int_{I_\rho} dv_{ii} \left| \rho^\prime (v_{ii}) \right|  \left \{   \left| \Psi_0^T  \widetilde{R}_{M+1,L,0}(z)  \Psi_0 - \Psi_0^T  \widetilde{R}_{M,L,0}(z)  \Psi_0  \right|^s  \right\} }  
&  \leq & C_{L,s,\rho}e^{-{\alpha_{L,s}}M} , \nonumber \\
 &   &    
\eea
for some finite positive constants uniform in $\C^+$.

\noindent
5. Therefore, combining the bounds \eqref{eq:smooth71} and \eqref{bound2}, we obtain the bound,
\bea\label{eq:telescopeSum2}
  \left| \frac{d}{dE} \E  \left\{ \langle e_0,  \widetilde{R}_{M+1,L,0}(z) e_0 \rangle \right\}  -  \frac{d}{dE} \E \left\{ \langle e_0,  \widetilde{R}_{M,L,0}(z)   e_0 \rangle \right\} \right|   & \leq  &  \sum_{i=-M-1}^{M+1} ~ C_{L,s,\rho} e^{- \alpha_{L,s} M}     \nonumber\\
     &\leq &  C_{L,s,\rho} (2M+1)e^{-\alpha_{L,s} M} .   \nonumber \\
  &  &
\eea 
From the telescoping series expansion \eqref{eq:telescoping} for $ \E \{ \Im\langle  e_0, (H_L^\infty-z)^{-1} e_0 \rangle \}$, we obtain the bound
        \begin{align}
    \left|\frac{d}{dz}\E  \left\{  \Im\langle  e_0, (H_L^\infty-z)^{-1} e_0 \rangle \right\} \right.& \left.-\frac{d}{dz}\E  
\left\{   \Im\langle e_0, (H_L^N-z)^{-1} e_0 \rangle \right\}    \right| \nonumber\\
    &\leq \sum_{M=N}^\infty C(1+2M)e^{-\alpha_{L,s} M}\nonumber\\
 &< \infty.
    \end{align}
Since this bound holds uniformly for $z \in \C^+$, we have smoothness of the density of states function by Lemma \ref{lemma:smoothCondition1}.
\end{proof}

As in \cite{dkm},  we can extend the methods above to prove the existence of a $k^{th}$-order derivative of the IDS when
the probability density $\rho$ satisfies Assumption 1 with index $k$, and, in particular, when $\rho$ is a Gaussian density as in \eqref{eq:gauss1}.  
The details of the proof of the following corollary are given in \cite{brodie1} and follow the ideas presented here and in \cite{dkm}.

%

\begin{corollary}\label{cor:anyDerivativeIDS1}
Let $H_L^\infty$ be a fixed-width random band matrix on $\ell^2(\Z)$ with random entries distributed with a probability density satisfying Assumption 1. Then the integrated density of states $N_L^\infty\in \mathcal{C}^{k}(\R)$.
\end{corollary}

\section{Local eigenvalue statistics for fixed-width RBM}\label{sec:poisson1}
\setcounter{equation}{0}

We recall that $\{E_L^N(j)\}_{j=-N}^N$ is the set of the $2N+1$ eigenvalues of $H_L^N$. To study the local eigenvalue statistics (LES) for $H_L^N$ around $E_0\in\R$, we define the re-scaled eigenvalues
\beq\label{eq:rescaledEV1}
\widetilde{E}_L^N(j):=(2N+1)\left(E_L^N(j) - E_0\right).
\eeq
The {\it local eigenvalue point process for $H_L^N$ centered at $E_0$} is a random point measure on $\R$ supported on the re-scaled eigenvalues:
\beq\label{eq:localEVptPr1}
\xi_{N,L}^\omega (s) ~ ds := \sum_{j=-N}^N \delta ( \widetilde{E}_L^N(j) - s) ~ ds.
\eeq

\begin{theorem}\label{thm:poisson1}
Let $H_L^N$ be a random band matrix with fixed bandwidth $L$ with matrix elements satisfying Assumption 1 and let
$\{E_L^N(j) \}_{j=-N}^N$ be the set of eigenvalues of $H_L^N$. The re-scaled eigenvalue point process
$\xi_{N,L}^\omega$, defined in \eqref{eq:localEVptPr1}
converges weakly to a Poisson point process with intensity measure given by $n_L^\infty (E_0)  ds$, the density of states function at $E_0$ times Lebesgue measure.
That is, for any bounded interval $A\subset \R$, the intensity of the limiting point process is
$$
\lim_{N\to\infty} \E\{\xi_{N,L}^\omega(A)\} = n_L^\infty (E_0)|A|,
$$
where $n_L^\infty (E)=\lim_{N\to\infty}n_L^N(E)$ is the uniform limit of the local density of states functions at $E$.
\end{theorem}

Given the local eigenvalues bounds of section \ref{sec:fundamentals1} and  the localization bounds of section \ref{sec:localBound1},  the proof of local eigenvalue statistics follows the same strategy as that of Minami \cite{minami} for the Anderson tight-binding model. In the first step, we use the Wegner and Minami estimates, and the localization bounds to prove that the process $\xi_{N,L}^\omega$ has the same weak limit points as a process  $\zeta_{N,L}^\omega$ constructed from an associated array of independent point processes as in \eqref{eq:sumPtProc1}. The second step consists of proving that the limit point of the process  $\zeta_{N,L}^\omega$ is unique and is a Poisson point process with the intensity measure $n_L^\infty (E_0)  ds$. 
The uniqueness relies critically on the Minami estimate. The identification of the intensity of the Poisson point process as $n_L^\infty (E_0)  ds$  requires control over the density of states as proven in Theorem \ref{thm:dosLimit1}.   We will sketch the proof of the first step as it is now rather standard. The second step, the study of the intensity measure of the limiting process is, however, different for random band matrices and relies  on Theorem \ref{thm:dosLimit1}.  This is presented in Proposition \ref{prop:Intensity1}. 


%
%
%
%


\subsection{Reduction to the point process $\zeta_{N,L}^\omega$}\label{subsec:array1}

We construct a point  process  $\zeta_{N,L}^\omega$ from an array of independent point processes and prove that it has the same limit points as $\xi_{N,L}^\omega$. To begin, we divide the set of indices $[-N,N]$ into sub-intervals of length $n$ where $n\sim N^\alpha$ for some $0<\alpha<1$.  We will label the subset of indices $N_p$ for $p=1, \dots, N^{1-\alpha}$.
When $N$ is large enough so that $N \gg n \gg L$, this produces $N^{1-\alpha}$ sub-matrices $H_L^{N,p}$ with band width $L$, and indices ranging over $N_p \times N_p$.  The sub-matrices are independent of each other, and thus the eigenvalue statistics of each sub-matrix are independent of the statistics of any other sub-matrix.  We first want to show that in the limit as $N\to\infty$, the eigenvalue point process of $H_L^N$ is approximated exactly by the sum of the point processes for $H_L^{N,p}$.

\begin{lemma}\label{lemma:reduction1}
Let $\{E_L^{N,p} (j)\}_{j\in N_p}$ denote the eigenvalues of $H_L^{N,p}$ and define the rescaled eigenvalues by 
$\widetilde{E}_L^{N,p} (j ):= (2N+1)(E_L^{N,p}(j) - E_0)$, for a fixed $E_0 \in \R$. We
define the local eigenvalue point process for $H_L^{N,p}$ by
$$
\xi_{N,L}^{\omega;p}(s)\,ds=\sum_{j\in N_p} \delta ( \widetilde{E}_L^{N,p}(j) -s) \,ds,  
$$
and let $\zeta_{N,L}^\omega$ be the sum of these independent point processes:
\beq\label{eq:sumPtProc1}
\zeta_{N,L}^\omega (s) ~ds :=  \sum_{p=1}^{\lfloor N - N^\alpha \rfloor}  \xi_{N,L}^{\omega;p}(s) ~ds. 
\eeq
Then for any bounded interval $A \subset \R$,
\beq\label{eq:evStat1} 
 \lim_{N \rightarrow \infty} | \E \{ \xi_{N,L}^\omega (A) \} - \E  \{ \zeta_{N,L}^\omega (A)  \} | = 0.
\eeq 
\end{lemma}

Note that the scaling for the sub-matrix point processes is the same as the scaling of the point process for the full matrix.
The details of the proof of Lemma \ref{lemma:reduction1} for fixed-width RBM is given in \cite{brodie1}, following the ideas in \cite{minami}. 

\subsection{Analysis of the point process $\zeta_{N,L}^\omega$}\label{subsec:arrayAnalysis1}

We recall the criteria in \cite[Volume 2, Chapter 11]{dvj}
for the limit point of an array of independent point processes to be a Poisson point process with intensity measure 
$\mu_{E_0}$.  For each bounded interval $A \in \R$,
\begin{enumerate}
		\item The array is uniformly asymptotically negligible:
		\begin{equation}
                     \lim_{N\to \infty}\sup_p \P \left\{ \xi_{N,L}^{\omega;p} (A)\geq 1 \right\} =0. 
                      \end{equation}
	
\item The limit point is unique:
		\begin{equation} 
                     \lim_{N\to\infty} \sum_p  \P  \left\{ \xi_{N,L}^{\omega;p}(A)\geq 2 \right\}=0. 
                      \end{equation}

		\item The intensity measure is identified via the convergence:
		\begin{equation}\label{eq:LimitDos1}
                      \lim_{N\to\infty} \sum_p \P  \left\{ \xi_{N,L}^{\omega;p}(A) \geq 1 \right\}= \mu_{E_0}(A).
                      \end{equation}
		
\end{enumerate}
In addition, we identify the intensity measure  as $d \mu_{E_0} (s) = n_L^\infty (E_0) ~ds$.

It is straight-forward to show that the Wegner estimate \eqref{eq:wegnerEst1} establishes point 1. 
Similarly, the Minami estimate \eqref{eq:minamiEst1} implies point 2.


\subsection{Identification of the intensity of the limiting process}

Before moving to the proof of point 3, we note the following consequence of the Minami estimate \eqref{eq:minamiEst1} which appeared \cite[(2.53)-(2.54)]{minami}.

\begin{lemma}\label{lemma:minami2}
Let $X_N$ be a sequence of random variables taking values in $\mathbb{N}$.  Further suppose $\E\{X_N(X_N-1)\}\to 0$ as $N\to \infty$.  Then
\begin{equation}\lim_{N\to\infty} |\E\{X_N\}-\P \{X_N\geq 1\}|=0. \end{equation}
\end{lemma}

This result is used in the proof of the following proposition establishing the Poisson nature of the limiting point process of $\zeta_{N,L}^{\omega}$. 



\begin{prop}\label{prop:Intensity1}
With the above definitions, the point process $\zeta_{N,L}^\omega$ 
converges weakly to a Poisson point process with intensity measure $n_L^\infty(E_0)dx$, where $n_L^\infty(E_0)$ is the uniform
 limit of the density of states functions evaluated at $E_0$.
\end{prop}

\begin{proof}
It remains to compute the limit in \eqref{eq:LimitDos1} and identify the limiting measure.  Suppose $A \subset \R$ is abounded interval. By Lemma \ref{lemma:minami2}, we can replace
$\P \{\xi_{N,L}^{\omega;p}(A) \geq 1\}$ with $\E\{\xi_{N,L}^{\omega;p}(A)\} $ in the $N\to\infty$ limit.
Further, from Lemma  \ref{lemma:reduction1} we can replace
$ \E \{ \zeta_{N,L}^\omega (A) \}$  
with $\E\left\{\xi_{N,L}^\omega (A) \right\} $ 
in the $N \rightarrow \infty$ limit. 
Then for $\varphi_z (u)  = \frac{1}{u-z}$, with $\Im z > 0$, 
\bea\label{eq:reduction2}
 \E\left\{\xi_{N,L}^{\omega}(\varphi_z)\right\} &  = &
\E\left\{ \sum_{j=-N}^N \delta_{(2N+1)\left(E_L^{N}(j)-E_0\right)}(\varphi_z) \right\}
 \nonumber \\
  & =  & \E \left\{ \tr \varphi_z\left( (2N+1)\left(H_{L}^{N}-E_0 \right)\right)
\right\} \nonumber \\
   & =   &  \sum_{j=-N}^N   \E \left\{  \Im \left\langle e_j , [ (2N+1)\left(H_{L}^{N}-E_0\right) - z ]^{-1} e_j \right\rangle
\right\}
\eea
Letting $z (N) := \frac{z}{2 N + 1}$, we obtain from \eqref{eq:reduction2}
\beq\label{eq:reduction3}
\E \left\{ \xi_{N,L}^\omega (\varphi_z)\right\}  = \int_\R  \Im \left( \frac{1}{x-E_0 - z(N) } \right) ~d\nu_L^N(x) ,
\eeq
where $\nu_L^N$ is the density of states measure for $H_L^N$. 
Using the fact that the $\ell$DOSm has a density $n_L^N$, after a change of variables, we obtain from \eqref{eq:reduction3}
\beq\label{eq:reduction4}
\E\left\{\xi_{N,L}^\omega (\varphi_z)\right\}  =
 \int \frac{1}{u^2 +1 } n_L^N\bigl( E_0 + u \Im z(N) + \Re z(N)\bigr)\, du.
\eeq
We now compute $\lim_{N \rightarrow \infty} \E\ \{ \xi_{N,L}^\omega (\varphi_z) \}$. The $\ell$DOSf  $n_L^N$ is pointwise uniformly bounded in $N$ by the density of states representation \eqref{eq:dosDefn1} and the spectral averaging estimate Proposition \ref{proposition:SpectralAve1}.  Thus, there exists a constant $C$ such that
\begin{equation}
    \frac{1}{u^2 +1 } n_L^N\left( E_0 + u \Im z(N) + \Re z(N)\right)\leq \frac{C}{u^2+1} \in L^1(\R).
\end{equation}
We can therefore apply the Dominated Convergence Theorem to bring the limit inside the integral in \eqref{eq:reduction4}, so
it remains to compute $\lim_{N\to\infty} n_L^N ( E_0 + u \Im z(N) + \Re z(N) )$. The uniform convergence of the density of states functions $n_L^N$ to $n_L^\infty$,  proven in Theorem \ref{thm:dosLimit1}, and the continuity of $n_L^\infty$, following from Theorem \ref{thm:smoothDos1}, imply that 
\begin{align}
    &\lim_{N\to\infty}\left|n_L^N\left( E_0 + u \Im z(N) + \Re z(N)\right)-n_L^\infty(E_0)\right|\nonumber\\
    &\leq \lim_{N\to\infty}\biggl(\left|n_L^N\left( E_0 + u \Im z(N) + \Re z(N)\right)-n_L^\infty\left( E_0 + u \Im z(N) + \Re z(N)\right)\right|\nonumber\\
    &\hspace{170pt}+\left|n_L^\infty\left( E_0 + u \Im z(N) + \Re z(N)\right)-n_L^\infty(E)\right|\biggr)\nonumber\\
    &=0.
\end{align}
Thus, evaluating the limit $N \rightarrow \infty$ of \eqref{eq:reduction4}, we have
\begin{equation}
 \lim_{N \rightarrow \infty} \E\left\{\xi_{N,L}^\omega (\varphi_z)\right\} =    \int \frac{1}{u^2+1} n_L^\infty(E_0)\,du =\pi n_L^\infty(E_0)
 =  \|\varphi_z\|_1  n_L^\infty(E_0),
\end{equation}
since $\pi = \|\varphi_z\|_1$.  Since the convergence is uniform for $z\in \C^+$, the convergence holds holds for characteristic functions of bounded intervals by a density argument.  This establishes condition (3) and identifies the intensity of the Poisson point process as stated after condition (3).
\end{proof}

\begin{appendices}

 \section{Appendix A: Basic identities}\label{app:identities1}
\setcounter{equation}{0}

\noindent
We make use of the following standard results of functional analysis.

\subsection{Resolvent and semi-group identities}
\begin{lemma}\label{lemma:ResolvInt1}
Let $A$ be a self-adjoint $N \times N$-matrix, $E\in\R$, and $\eps>0$.  Then
$$
 ( A-E-i \eps)^{-1} = i \int_0^\infty e^{-i \lambda (A-E-i \eps) }\, d\lambda. 
$$
\end{lemma}
\begin{proof}
The result follows from direct integration after diagonalization.
\end{proof}

\begin{lemma}\label{lemma:duhamel1}
Let $A$ and $B$ be two $N \times N$-matrices. For any $t \geq 0$, we have 
\[ e^{itA}-e^{itB}  = i \int_0^t e^{i(t-s)A}(A-B)e^{isB}\,ds. \]
\end{lemma}
\begin{proof}
We write 
\[e^{itA}-e^{itB}=e^{itA}
\left(1-e^{-itA}e^{it B)}\right)
\]
and apply the fundamental theorem of calculus.
\end{proof}

\noindent
We say that a $N \times N$-matrix $A$ has positive imaginary part if $\Im A \geq 0$. 

\begin{lemma}\label{lemma:duhamel2}
Let $A$ and $B$ be two normal $N \times N$-matrices with positive imaginary parts.  Then for each $0 \leq s \leq 1$ and for any $t \geq 0$, we have
\[ \left\|e^{itA}-e^{itB}\right\| \leq 2^{1-s} |t|^s \|A-B\|^s.\]
\end{lemma}

\begin{proof}
Since  $A$ and $B$ are normal and their imaginary parts are positive, we have
$\| e^{i t A}\|, \|e^{it B} \| \leq 1$, for $t \geq 0$. 
Then, Lemma \ref{lemma:duhamel1}  results in the bound
\[ \left\| e^{itA}- e^{itB}\right\| \leq |t|\|A-B\|.  \]
Furthermore, by the triangle inequality, we have
\[ \left\| e^{itA}- e^{itB}\right\| \leq 2.\]
Combining these two estimates for any $0 \leq s \leq 1$, we get
\[\left\| e^{itA}- e^{itB}\right\|
=\left\| e^{itA}- e^{itB}\right\|^{1-s} \left\| e^{itA}- e^{itB}\right\|^s
\leq 2^{1-s} |t|^s\|A-B\|^s.\]
\end{proof}

\begin{remark}
If $A$ and $B$ generate contractive semigroups, so that $\| e^{itA}\|, \|e^{i t B} \|  \leq 1$, then Lemma \ref{lemma:duhamel2} holds. This version is used in \cite{dkm}.
\end{remark}

\subsection{The Schur complement formula}\label{subsec:schur1} 

Let $M : \C^N \rightarrow \C^N$ be a linear transformation and suppose $P$ is an orthogonal projection on $\C^N$ with $Q := 1 - P$.
Relative to $P$ and $Q$, we write $M$ as a square matrix in block form as
\[ M= \left(\begin{array}{cc}
    A & B \\
    C & D
    \end{array}\right) =  \left(\begin{array}{cc}
    PMP & PMQ \\
    QMP & QMQ
    \end{array}\right)   . \]

\begin{lemma}\label{lemma:schur1} 
Suppose $M$ is invertible, and $D = Q M Q$ is invertible on its range $Q \C^N$. We then have
\[ P M^{-1} P =\left(PMP - P M Q (Q M Q)^{-1} Q M P\right)^{-1} \]
\[=\left( A- B D^{-1} C\right)^{-1}.\]
\end{lemma}

\begin{proof}
We first note that under the hypotheses, $M$ admits the following factorization:
\[ M = \left( \begin{array}{cc}
                A & B \\
                C & D \end{array}\right)
        =\left(\begin{array}{cc}
                I    &  BD^{-1} \\
                0    &  I \end{array}\right)
        \left(\begin{array}{cc}
                A-BD^{-1}C     & 0  \\
                0    &  D \end{array}\right)
        \left(\begin{array}{cc}
                I    &  0 \\
                D^{-1}C    &  I \end{array}\right).
\]
The first and third factors are invertible since $D$ is assumed to be invertible. Since $M$ is assumed to be  invertible,  so is the second factor. We then have
\[ M^{-1} =
      \left(\begin{array}{cc}
                I    &  0 \\
                D^{-1}C    &  I \end{array}\right)^{-1}
        \left(\begin{array}{cc}
                A-BD^{-1}C     & 0  \\
                0    &  D \end{array}\right)^{-1}
        \left(\begin{array}{cc}
                I    &  BD^{-1} \\
                0    &  I \end{array}\right)^{-1}
\]
\[
=      \left(\begin{array}{cc}
                I    &  0 \\
                -D^{-1}C    &  I \end{array}\right)
        \left(\begin{array}{cc}
                \left(A-BD^{-1}C\right)^{-1}     & 0  \\
                0    &  D^{-1} \end{array}\right)^{-1}
        \left(\begin{array}{cc}
                I    &  -BD^{-1} \\
                0    &  I \end{array}\right)^{-1}
\]
\[
=\left(\begin{array}{cc}
    \left(A-BD^{-1}C\right)^{-1} & -\left(A-BD^{-1}C\right)^{-1}BD^{-1} \\
    -D^{-1}C\left(A-BD^{-1}C\right)^{-1} & D^{-1} C \left(A-BD^{-1}C\right)^{-1} B D^{-1} + D^{-1}
\end{array}\right).
\]
\end{proof}

The matrix $A - B D^{-1} C$ is called the\emph{ Schur complement} of $D$ in $M$ and is invertible if $M$ and $D$ are invertible.


\section{Appendix B: Localization for RBM at high energies}\label{sec:loc1}
\setcounter{equation}{0}

The localization bounds of Schenker, Theorem \ref{thm:fractMoment1}, hold on a finite, but arbitrary, real energy interval.
In order to extend these to all energies, we use the Aizenman-Molchanov method \cite{AM} to prove Theorem \ref{thm:locHE1}. This leads to the uniform exponential decay estimate of Theorem \ref{thm:locAllE1}.


\subsection{Proof of Theorem \ref{thm:locHE1}}

\noindent
1. We first note that we have the following {\it a priori} bound which is stated in part (i) of Proposition \ref{proposition:SpectralAve1}:
\[ 
\E \{ | \langle e_j, (H_L^N-z )^{-1} e_\ell \rangle |^s   \} \leq C_{\rho,s}
\]
for some constant uniform in $j, \ell$, the size $N$, and the energy $z$.
To derive the fractional moment bound, we begin with the identity
\[ \langle e_j, (H_L^N - z  ) (H_L^N - z )^{-1} e_\ell \rangle = \delta_{j \ell}. \]
Thus for $j \neq \ell$, we have
\[ \sum_{m\,:\, |j-m|\leq L} \left(v_{jm} - z \delta_{jm}\right)\left\langle
e_{m}, \left(H_L^N - z \right)^{-1} e_\ell \right\rangle = 0.
\]
where $v_{jm} = [H_L^N]_{jm}$.
Rearranging the sum and extracting the $jj^{\rm th}$-term gives
\[\left(v_{jj} - z \right)\left\langle
e_{j}, \left(H_L^N - z \right)^{-1} e_\ell \right\rangle
=
-\sum_{\substack{ m\,:\, |m - j|\leq L \\
m \neq j} }  v_{jm}\left\langle
e_{m}, \left(H_L^N - z \right)^{-1} e_\ell \right\rangle.
\]
By means of the identity  \eqref{eq:BasicIneq1}, the expectation of the $s^{\rm th}$-power of this equaltion results in the inequality
\begin{equation}\label{eq:AMbasic1}
    \E\left\{\left| v_{jj} - z \right|^s \left| \left\langle
e_{j}, \left(H_L^N - z \right)^{-1} e_\ell \right\rangle\right|^s\right\}
\leq 
\sum_{\substack{ m\,:\, |m-j|\leq L \\
m \neq j} }  \E\left\{ \left| v_{jm}\right|^s \left|\left\langle
e_{m}, \left(H_L^N - z \right)^{-1} e_\ell \right\rangle\right|^s\right\}.\end{equation}


\noindent
2. By the Schur complement formula, the left side of \eqref{eq:AMbasic1}  may be written in the form
\beq\label{eq:lowerDecoup2} 
 |v_{jj} - z|^s \left|\left\langle
e_j, \left(H_L^N - z \right)^{-1} e_\ell \right\rangle \right|^s =  |v_{jj} - z|^s  \left| \frac{ A}{B(v_{jj}-z) + C}\right|^s ,
\eeq   
where $A$, $B$,  and $C$ are independent of $v_{jj}$.  Thus Lemma \ref{lemma:lowerdecoup}, the lower decoupling estimate, provides a lower bound on the expectation of \eqref{eq:lowerDecoup2} of the form
\beq\label{eq:LowerDecoup1}
 C\left(|z|\right)^s \E \left\{\left| \left\langle
e_{j}, \left(H_L^N - z \right)^{-1} e_\ell \right\rangle\right|^s\right\},
\eeq
where $C\left(|z|\right)$ scales like $|z|$ for large $|z|$.
As for the right side of  \eqref{eq:AMbasic1},  we can also use the Schur complement formula with $P=P_{m}$ to write
\begin{equation} \left\langle
e_{m}, \left(H_L^N - z \right)^{-1} e_\ell \right\rangle = \frac{Av_{jm}+B}{C{v_{jm}}^2+Dv_{jm}+E}
\end{equation}
for $j\neq m$, where $A$, $B$, $C$, $D$, and $E$ are all independent of $v_{jm}$.  Thus we can use Lemma \ref{lemma:upperdecoup}, the upper decoupling bound,  on each term in the sum on the right side of \eqref{eq:AMbasic1} to obtain the bound
\begin{align}\label{eq:UpperDecoup1}
\E\left\{ \left| v_{jj'}\right|^s \left|\left\langle 
e_{m}, \left(H_L^N - z \right)^{-1} e_\ell \right\rangle\right|^s\right\}&=
\E\left\{ \left| v_{jm}\right|^s \left|\frac{Av_{jm}+B}{C{v_{jm}}^2+Dv_{jm}+E}\right|^s\right\}
\nonumber\\
&\leq C\E\left\{ \left|\frac{Av_{jm}+B}{C{v_{jm}}^2+Dv_{jm}+E}\right|^s \right\}\nonumber\\
&= C \E\left\{ \left|\left\langle 
e_{m}, \left(H_L^N - z \right)^{-1} e_\ell \right\rangle\right|^s\right\}.
\end{align}

\noindent
3. Rewriting \eqref{eq:AMbasic1} with using \eqref{eq:LowerDecoup1} and \eqref{eq:UpperDecoup1}, the new lower and upper bounds, gives 
\begin{equation}\label{eq:AMbasic2}
\E \left\{\left| \left\langle
e_{j}, \left(H_L^N - z \right)^{-1} e_\ell \right\rangle\right|^s\right\}
\leq \frac{C}{C(|z|)^s} 
\sum_{\substack{ m\,:\, |m-j|\leq L \\
m \neq j} } \E\left\{ \left|\left\langle 
e_{m}, \left(H_L^N - z \right)^{-1} e_\ell \right\rangle\right|^s\right\}.
\end{equation}
This estimate can now be iterated $|j-\ell|/L$ times until at least one of the indices in the sum overlaps with $\ell$. 
 At the end of the iteration process, we use the \textit{a priori} bound \eqref{eq:matrixEle1} 
to bound each expectation by an absolute constant.  Thus, we have 
\begin{equation}
\E \left\{\left| \left\langle
e_{j}, \left(H_L^N - z \right)^{-1} e_\ell \right\rangle\right|^s\right\}
\leq \left(\frac{2LC}{C(|z|)^s}\right)^\frac{|j - \ell|}{L}.\end{equation}
We now take $R$ such that $C(R)^s> 2L C$ and note that for each $z$ with $|z|>R$, we have exponential decay.
\hfill$\square$

\subsection{Decoupling lemmas}

To derive a fractional moment bound for random band matrices which have both diagonal and off-diagonal randomness, we need both upper and lower decoupling estimates.  
The following lemma is contained in \cite[Lemma 3.1]{AM}, see also \cite{graf} for a similar result.  

\begin{lemma}[Lower Decoupling]\label{lemma:lowerdecoup}
Let $\rho$ be the density of a Lipschitz continuous probability measure and $0<s<1$.  Then there is a $C>0$ such that 
\begin{equation} \int  \frac{|v-\eta|^s }{ |v-\beta|^s }\rho(v)\,dv \geq C_\rho(|\eta|)^s \int \frac{1}{|v-\beta|^s} \rho(v)\, dv.
\end{equation}
If $\int |v|^\gamma \rho(v)\, dv < \infty$ for some $\gamma>s$, then $C_\rho(R)^s$ is an increasing function of $R>0$ and 
\begin{equation}\lim_{R\to\infty} \frac{C_\rho(R)}{R} = 1.
\end{equation}
\end{lemma}
In \cite{AM}, the authors require less strict conditions on the probability density $\rho$.  They require $\rho$ to be locally uniformly $\tau$-H\"{o}lder continuous, and obtain constants that depend on $\tau$. If $\rho$ is Lipschitz, we can take $\tau=1$.

The general upper decoupling estimate appears in \cite[Theorem III.2]{AM}. We use the following version tailored to our application. 

\begin{lemma}[Upper Decoupling] \label{lemma:upperdecoup}
Let $\rho$ be the density of a Lipschitz continuous probability measure on $\R$ with  a finite $\gamma$-moment condition
$\int |u|^\gamma\, \rho(u)\, du < \infty$, for some $\gamma > 0$, and let $0<s<1$.  Then, for any polynomials 
$p$ of degree $n$ and $q$ of degree $k$ satisfying $s(n+k) \leq \gamma$ and $s k<1$, there exists a finite $C > 0$, depending on  $(\gamma,s,n,k)$ and $\rho$, so that
\begin{equation}\label{eq:UpperDecoup0}
\int_\R |u|^{\gamma-s(n+k)}
\widetilde{\rho}(u)\, du \leq C\int_\R
|u|^\gamma {\rho}(u)\,du  ,
\end{equation}
where $\widetilde{\rho}$ is the probability density
$$
\widetilde{\rho}(u) :=  \left({\int \frac{|p(v)|^s}{|q(v)|^s} \rho (v) ~dv} \right)^{-1}~ \frac{ |p(u)|^s}{|q(u)|^s }~\rho (u) .
$$
\end{lemma}
%
%
%

In the application of the upper decoupling estimate in \eqref{eq:UpperDecoup0}, we have  $n=1$ and $k=2$
and we want $\gamma - s(n+k) = s$ so that $\gamma = 4s$ and $s < \frac{1}{k} = \frac{1}{2}$. This requires that the probability density have a finite $\gamma$-moment for some $0< \gamma < 2$.  
 For Gaussian random variables, we may take $\gamma$ to be any positive real number.  Under these conditions, 
the upper decoupling bound \eqref{eq:UpperDecoup0} becomes:
\beq\label{eq:UpperDecoup2}
 \int_\R |u|^{s}
\frac{|p(u)|^s}{|q(u)|^s} 
\rho(u)\, du 
\leq C\int_\R \frac{|p(u)|^s}{|q(u)|^s} 
 \rho(u)\,du ,  
\eeq
which is required in order to obtain the upper bound in \eqref{eq:UpperDecoup1}. 

\end{appendices}


\end{document}